\documentclass[journal]{IEEEtran}

\usepackage{multicol}
\usepackage{multirow}
\usepackage{longtable}
\usepackage{subcaption}

\usepackage{balance}
\usepackage{cite}

\usepackage{algorithm}
\usepackage{algpseudocode}

\ifCLASSINFOpdf
   \usepackage[pdftex]{graphicx}

\DeclareGraphicsExtensions{.pdf,.jpeg,.png}
\else
 
\fi

\usepackage{mathrsfs,amsmath}
\usepackage{amssymb}

\usepackage{amsthm,hyperref,enumerate}
\newtheorem{theorem}{Theorem}

\theoremstyle{definition}
\newtheorem{defn}{Definition}

\theoremstyle{remark}
\newtheorem{lemma}{Lemma}

\theoremstyle{plain}
\newtheorem{remark}{Remark}
 \newtheorem{assumption}{Assumption}

\usepackage{ragged2e}

\usepackage[amssymb]{SIunits}
\hypersetup{
	colorlinks=true,
	linkcolor=blue,
	filecolor=blue,      
	urlcolor=blue,
}

\begin{document}

\title{Distributed Semi-global Output Feedback Formation Maneuver Control of High-order Multi-agent Systems}

\author{Xu Fang, Lihua Xie,~\IEEEmembership{Fellow,~IEEE}
\thanks{This work was supported by ST Engineering-NTU Corporate Lab under its NRF Corporate Lab @ University Scheme (Singapore) and Projects of Major International (Regional) Joint Research Program under NSFC (China) Grant no. 61720106011. (Corresponding author: Lihua Xie.)}
\thanks{
X. Fang and L. Xie are with the School of Electrical and Electronic Engineering, Nanyang Technological University, Singapore. (E-mail: fa0001xu@e.ntu.edu.sg; 
elhxie@ntu.edu.sg).}
}

\maketitle

\begin{abstract}
This paper addresses the formation maneuver control problem of leader-follower multi-agent systems with high-order integrator dynamics. 
A distributed output feedback formation maneuver controller is proposed to achieve desired maneuvers so that
the scale, orientation, translation, and shape of formation can be manipulated continuously, where the followers do not need to know or estimate the time-varying maneuver parameters only known to the leaders. Compared with existing
relative-measurement-based formation maneuver control, the advantages of the proposed method are that it is output (relative output) feedback based and shows how to realize different types of formation shape. In addition, it can be applied to non-generic and non-convex nominal configurations and the leaders are allowed to be maneuvered. It is worth noting that
the proposed method can 
also be extended to general linear multi-agent systems under some additional conditions. 
The theoretical results are demonstrated by a simulation example.
\end{abstract}


\IEEEpeerreviewmaketitle

\section{Introduction}

\IEEEPARstart{M}{ulti-agent} formation control 
is applicable to the field of engineering such as load transportation, cooperative localization, and area exploration \cite{  wang2018distributed, deng2019cooperative,  wang2017distributed}. The formation shape control steers mobile agents to a desired time-invariant formation \cite{ suttner2018formation, li2021adaptive}. Different from formation shape control,
the formation maneuver control aims to maneuver the entire group of agents to change the scale, orientation, translation, and shape of formation continuously.

To tackle the formation maneuver control problem, mobile agents are usually divided into leader group and follower group. The leaders  determine the time-varying maneuver parameters such as scale, orientation, translation, and shape of formation, while the followers try to form the desired formation. 
To form a target time-varying formation, the existing methods are divided into two categories: (\text{\romannumeral1})
Each agent can access the  maneuver parameters;  (\text{\romannumeral2}) Only the leaders have access to the maneuver parameters. When each agent can access the maneuver parameters, the formation maneuver control becomes a consensus tracking or containment control problem \cite{wen2016consensus, zhang2020time,dong2017time}, where agents collaboratively keep a time-varying formation. For example, in the containment-control-based time-varying formation \cite{dong2017time}, all agents know the time-varying maneuver parameters.
In some application scenarios, only the leaders have access to
the time-varying maneuver parameters. For this case,
estimator can be designed for each follower to obtain the formation maneuver parameters \cite{yang2019stress, park2015formation}. However,
the performance of the controllers
will be influenced by the estimation errors, and the estimators consume computing and communication resources. To reduce consumption of computing and communication resources,
the second approach is to explore the
invariance of inter-agent constraints, such as distances and relative bearings, of the desired formation. The distance-based formation maneuver control can realize the rotational and translational maneuvers \cite{de2016distributed}. Similarly, the scaling and translational maneuvers can be achieved by bearing-based methods \cite{zhao2015translational}.

To change the scale, orientation, and translation of formation
without the need of designing the estimators for the followers to estimate the time-varying maneuver parameters only known to the leaders, some estimator-free-based methods are proposed recently, e.g,   barycentric coordinate \cite{ han2017fobarycentric, han2015three}, stress matrix \cite{zhao2018affine, li2020layered, chen2020distributed, zhi2021leader,  xu2018affinel,  xu2020affine, xu2020affinede, lin2021unified,zhou2022distributed,zhu2022completely,gao2022practical,zhu2022distributed}, and complex Laplacian \cite{han2015formation,lin2014distributed}. 
However, there are some limitations in the existing estimator-free-based methods \cite{zhao2018affine, li2020layered, chen2020distributed, xu2018affinel,  xu2020affine, han2017fobarycentric,  han2015formation,lin2014distributed, han2015three, zhi2021leader, lin2021unified, xu2020affinede,zhou2022distributed,zhu2022completely,gao2022practical,zhu2022distributed}. First, 
the complex-Laplacian-based methods are only applicable to 2-D space \cite{lin2014distributed,han2015formation}. Second, the works in \cite{zhi2021leader, lin2021unified, zhao2018affine,  xu2020affinede,  han2017fobarycentric, han2015three, lin2014distributed,han2015formation, li2020layered,zhou2022distributed,zhu2022completely,gao2022practical,zhu2022distributed } can only be applied in low-order or Euler–Lagrange multi-agent systems.
Third, the approaches in \cite{zhao2018affine, li2020layered, chen2020distributed, xu2018affinel,  xu2020affine,han2017fobarycentric, han2015formation, han2015three, zhi2021leader, lin2021unified, xu2020affinede,zhou2022distributed,zhu2022completely,gao2022practical,zhu2022distributed} require a convex or generic nominal configuration, i.e., the nominal followers are in a convex hull formed by the nominal leaders, or any three nominal agents are not on the same line and any four nominal agents are not on the same plane. For instance, the affine localizability of a nominal formation is guaranteed if the nominal configuration of formation is generic \cite{zhao2018affine, chen2020distributed, xu2018affinel,xu2020affine,xu2020affinede, li2020layered, zhi2021leader, lin2021unified,zhou2022distributed,zhu2022completely,gao2022practical,zhu2022distributed}. Although it is conjectured in \cite{zhao2018affine}
that the affine localizability of a nominal formation may also be guaranteed by non-generic configurations, 
no theoretical proof is given.  Fourth, the existing results \cite{zhao2018affine, li2020layered, chen2020distributed, xu2018affinel,  xu2020affine, han2017fobarycentric,  han2015formation, han2015three, zhi2021leader, lin2021unified, xu2020affinede,lin2014distributed,zhou2022distributed,zhu2022completely,gao2022practical,zhu2022distributed} require the relative motion states rather than the measured (relative) outputs to achieve their control objectives.
Fifth, there is no controller for the leaders in
\cite{xu2018affinel, xu2020affine, xu2020affinede, han2017fobarycentric,zhu2022completely,gao2022practical,zhu2022distributed}, which limits the maneuverability of formation, and the tracking errors of the formation maneuver control in \cite{xu2018affinel} can only converge to a bounded
compact set if the positions of the leaders are time-varying.

"High-order integrator dynamics" in this article refer to the $m$th order integrator dynamics, where $m$ can be any positive integer.
The work in \cite{fang2020tc} is only applicable to low-order multi-agent systems, which cannot be trivially extended to  multi-agent systems with higher-order integrator dynamics.  Further, it
requires the relative motion states of the agents. Although the relative positions may be measured by vision technology,
it is difficult to measure the relative velocities and
relative accelerations in real applications, let alone the higher-order relative motion states in higher-order multi-agent systems. Motivated by its limitations,  a distributed algorithm is proposed that can use outputs or relative outputs to achieve formation maneuver control for multi-agent systems with high-order integrator dynamics. Compared with existing results in \cite{zhao2018affine, li2020layered, chen2020distributed, xu2018affinel,  xu2020affine, han2017fobarycentric,  han2015formation, han2015three, zhi2021leader, lin2021unified, xu2020affinede, fang2020tc,lin2014distributed,zhou2022distributed,zhu2022completely,gao2022practical,zhu2022distributed} which require the relative motion states to achieve their control objectives, the first advantage is that it is output (relative output) feedback based approach. The second advantage is that it is applicable to both non-generic and non-convex nominal configurations, where the proposed theory is different from that in \cite{zhao2018affine, li2020layered, chen2020distributed, xu2018affinel,  xu2020affine, han2017fobarycentric,  han2015formation, han2015three, zhi2021leader, lin2021unified, xu2020affinede, fang2020tc,lin2014distributed,zhou2022distributed,zhu2022completely,gao2022practical,zhu2022distributed}. Although the shape change of formation is discussed in the affine-theory-based methods \cite{zhao2018affine, chen2020distributed, xu2018affinel,xu2020affine,xu2020affinede, li2020layered, zhi2021leader, lin2021unified,zhou2022distributed,zhu2022completely,gao2022practical,zhu2022distributed}, their shape change is usually "shearing and colinearity". Thus,
the third advantage is that the proposed method shows how to realize different types of formation shape by only tuning a formation shape parameter. In addition, the proposed method can be extended to general linear multi-agent systems under some additional conditions. It is worth noting that the conventional output feedback formation control methods \cite{ wang2018distributed,deng2019cooperative,wang2017distributed} with Laplacian matrix cannot be trivially extended to the discussed case. The edge weights of a conventional Laplacian matrix are
generally nonnegative, but the edge weights in the proposed
method are designed to be either negative or nonnegative. The main contributions are summarized as below:
\begin{enumerate}[(i)]
\item A rigidity-theory-based distributed semi-global output feedback formation maneuver control algorithm for multi-agent systems with high-order integrator dynamics is proposed, where the followers do not need to know or estimate the time-varying maneuver parameters only known to the leaders. 
\item The proposed method is applicable to both non-generic and non-convex nominal configurations, and allows leaders to maneuver to achieve the desired formation. 
\item The proposed method can
be applied in the case when only the relative outputs are available for the followers and can also be extended to general linear multi-agent systems under some additional conditions.
\item The proposed method shows how to realize different types of formation shape by only tuning a formation shape parameter. The closed-loop tracking errors of the formation maneuver control semi-globally and  asymptotically converge to zero.
\end{enumerate}

The remaining parts of the paper are organized as follows. Preliminaries and problem statement are given in Section \ref{preli}. In Section \ref{leader}, the output feedback formation maneuver control algorithms of the leaders and followers are presented. In Section \ref{discu},
the proposed method is compared with existing formation maneuver control of multi-agent systems with high-order integrator dynamics.  Section \ref{simu} presents
a numerical example to illustrate the theoretical findings and conclusions are drawn in 
Section \ref{conc}.

\section{Preliminaries and Problem Statement}\label{preli}

Let $\| \cdot \|_1$ and $\| \cdot \|_2$ be the $1$-norm and $2$-norm of a given matrix or vector, respectively.
${I}_d $ represents the $d \times d$ identity matrix. Let  ${\mathbf{1}}_d$ be the all-one vector in $\mathbb{R}^d$. Let ${\mathbf{0}}$ be the all-zero matrix or vector.
For a square matrix $\mathcal{A}$, $\mathcal{A}\!>\!0 \ ( \mathcal{A}\! < \!0)$ means that the matrix $\mathcal{A}$ is positive definite (negative definite).
$\lambda_{\max}(\mathcal{A})$ and $\lambda_{\min}(\mathcal{A})$  represent the  maximum and minimum eigenvalue of $\mathcal{A}$, respectively. 
Let $\otimes$ be the Kronecker product.  Denote $Q \in \text{SO}(d)$ as the set of $d$-dimensional rotation matrices. A matrix is Hurwitz if all the eigenvalues of the matrix have negative real parts.
Consider a formation of $n$ mobile agents in 3-D space whose positions are denoted by $p_i \in \mathbb{R}^3, i =1, \cdots, n$. A graph is denoted by $\mathcal{G}=\{ \mathcal{V},\mathcal{E}\}$ consisting of a non-empty  node set $\mathcal{V}=\{1,  \cdots, n \}$ and an edge set $\mathcal{E}  \subseteq \mathcal{V} \times \mathcal{V}$.  The neighbor set of agent $i$ is denoted by
$\mathcal{N}_i \triangleq \{ j \in \mathcal{V} : (i,j) \in \mathcal{E} \}$. A formation is represented by $(\mathcal{G}, p)$, where $p=[p_1^{\top},\cdots, p_n^{\top}]^{\top}$ is a configuration of the formation. In this article, the agents who
have access to the maneuver parameters are called leaders, while the agents who have no access to the maneuver parameters are called followers.
Consider a nonlinear system \cite{paden1987calculus}:
\begin{equation}
\dot{x}=f(x,t),    
\end{equation}
where $f(\cdot ):\mathbb{R}^{d}\times\mathbb{R}\rightarrow\mathbb{R}^{d}$ is a Lebesgue measurable and locally bounded discontinuous function. $x(\cdot)$ is called a Filippov solution of $f(x,t)$ on $[t_{0},t_{1}]$ if $x(\cdot)$ is absolutely continuous over $[t_{0},t_{1}]$ and satisfies the differential inclusion $x\in \mathcal{K}[f](x,t)$ for almost all $t\in \lbrack t_{0},t_{1}]$ with $\mathcal{K}[f]:=\cap _{\delta >0}\cap _{\mu (\bar{N})=0}\bar{co}(f(B(x,\delta ) \!-\! \bar{N}),t)$. $B(x,\delta )$  is the open ball of radius $\delta$  centered at $x$.  $\bar{co}({N})$ is the convex closure of set $N$. $\cap _{\mu (\bar{N}=0)}$ is the intersection over all sets $\bar{N}$ of Lebesgue measure zero. The convex hull is denoted by
$co(\cdot)$. The Clarke's generalized gradient of a locally Lipschitz continuous function $V$ is given by
$\partial V \! \triangleq \! co\{\lim \triangledown V(x_{i})|x_{i}\rightarrow x,x_{i}\in \Omega _{v}\cup \bar{N}\}$. $\Omega _{v}$ is the Lebesgue measure zero set where $\triangledown V(x_{i})$ does not exist.  $\bar{N}$ is a zero measure set \cite{paden1987calculus}. Then, the set-valued Lie derivative of $V$ with respect to $\dot{x}=f(x,t)$ is given by $\dot{\tilde{V}}:=\cap _{\varphi \in\partial V}\varphi ^{\top}\mathcal{K}[f](x,t)$.

\subsection{Leader-follower Formation}\label{cons}

Without loss of generality,
for a group of $n=n_l+n_f$ agents consisting of $n_l$ leaders and $n_f$ followers in 3-D space, denote $\mathcal{V}_l=\{1, \cdots, n_l\}$ as the leader set, and $\mathcal{V}_f=\{n_l\!+\!1, \cdots, n\}$ as the follower set. Let $r_l=[r_1^{\top}, \cdots, r_{n_l}^{\top}]^{\top}$ and $r_f=[r_{n_l\!+\!1}^{\top}, \cdots, r_n^{\top}]^{\top}$, where $r_i \in \mathbb{R}^3 , i \in \{1, \cdots, n \}$ is the nominal position of agent $i$.
A
nominal formation is denoted by $(\mathcal{G}, r)$, where $r=[r_l^{\top}, r_f^{\top}]^{\top}$ is the nominal configuration. Let $e_{ij} \! \triangleq \! r_i -  r_j$. A \textit{displacement parameter set} of agent $i$ is a set of scalars $\{ w_{ij} \}_{(i,j) \in \mathcal{E}}$ satisfying
\begin{equation}\label{root}
     \sum\limits_{j \in \mathcal{N}_i} w_{ij}e_{ij}= \mathbf{0}, \ \ \sum\limits_{j \in \mathcal{N}_i} \|w_{ij}\|_2 >0.
\end{equation}

The equation \eqref{root} is called a \textit{displacement constraint}, which can be rewritten as
\begin{equation}
    \sum\limits_{j \in \mathcal{N}_i} w_{ij}r_i- \sum\limits_{j \in \mathcal{N}_i}w_{ij}r_j\!=\! \mathbf{0}.
\end{equation}

The minimum number of agents required to construct a displacement constraint is given below. 

(\romannumeral1) If agent $i$ and any four neighboring agents
$j, k, h, l$ are not coplanar,  there is a non-zero vector $ w_i=(w_{ij},  w_{ik}, w_{ih},  w_{il})^{\top} \in \mathbb{R}^{4}$ such that 
\begin{equation}\label{en1}
 w_{ij}e_{ij}+ w_{ik}e_{ik}+w_{ih}e_{ih}+w_{il}e_{il}= \mathbf{0}. 
\end{equation}

(\romannumeral2) If agent $i$ and any three neighboring agents $j,k,h$ are coplanar, there is a non-zero vector $ \bar w_i=(w_{ij},  w_{ik}, w_{ih})^{\top} \in \mathbb{R}^{3}$ such that
\begin{equation}\label{en2}
 w_{ij}e_{ij}+ w_{ik}e_{ik}+w_{ih}e_{ih}= \mathbf{0}. 
\end{equation}

(\romannumeral3) If agent $i$ and its any two neighboring agents $j,k$ are colinear, there is a non-zero vector $\widehat w_i=(w_{ij},  w_{ik})^{\top} \in \mathbb{R}^{2}$ such that
\begin{equation}\label{en3}
 w_{ij}e_{ij}+ w_{ik}e_{ik}= \mathbf{0}. 
\end{equation}

The parameters $w_{ij},  w_{ik}, w_{ih},  w_{il}$ in \eqref{en1} are designed by solving the following matrix equation
\begin{equation}\label{element}
\left[ \!
\begin{array}{c c c c}
e_{ij} & e_{ik} & e_{ih}  &  e_{il} \\
\end{array}
\right]  \left[ \!
	\begin{array}{c}
	w_{ij} \\
	w_{ik} \\
	w_{ih} \\
	w_{il}
	\end{array}
	\right] = \mathbf{0}.
\end{equation}

Similarly, the parameters in \eqref{en2} and \eqref{en3} can also be designed. Each follower $i \! \in \! \mathcal{V}_f$ can form a displacement constraint \eqref{root}. A set of $n_f$ displacement constraints can be rewritten in a compact form shown as
\begin{equation}\label{form}
(\Omega_f \otimes I_3)r=\mathbf{0},  
\end{equation}
where $\Omega_f \in \mathbb{R}^{n_f \times n}$ is called the follower matrix satisfying
\begin{equation}\label{zer}
\begin{array}{ll}
    &[\Omega_f]_{ij} \!=\!  \left\{ \! \begin{array}{lll} 
    -w_{ij}, & 
    j \in \mathcal{N}_i, \ \ i \neq j, \\
    \ \  0, &  j \notin \mathcal{N}_i, \ \ i \neq j, \\
    \sum\limits_{k \in \mathcal{N}_i}w_{ik},
    & i=j. \\
    \end{array}\right. 
\end{array} 
\end{equation}

\begin{remark}\label{rre1}
The follower matrix can be non-square and non-symmetric and is different from the stress matrix defined in \cite{zhao2018affine}, which is a square and symmetric matrix, i.e., $w_{ij}=w_{ji}$. Moreover, the edge weights of 
a traditional Laplacian matrix are generally nonnegative except for some special cases; see, e.g., \cite{altafini2012consensus}.
Note also that the Laplacian matrix in \cite{altafini2012consensus} is square and symmetric. 
\end{remark}

Since the agents
are divided into leaders and followers, the follower matrix $\Omega_f$ can be partitioned as
\begin{align}\label{aef}
 \Omega_f=[\begin{array}{ll}
    \Omega_{fl}  &  \Omega_{f\!f} 
    \end{array}],
\end{align}
where $ \Omega_{fl} \in \mathbb{R}^{{n_f} \times {n_l}}$, $\Omega_{f\!f} \in \mathbb{R}^{{n_f} \times {n_f}}$. Then, \eqref{form} becomes
\begin{equation}\label{1aef}
   (\Omega_{fl} \otimes I_3)r_l +    (\Omega_{f\!f} \otimes I_3)r_f = \mathbf{0}.
\end{equation}

\begin{defn}
\cite{zhao2018affine} A nominal formation $(\mathcal{G}, r)$ is called localizable if  $r_f$ can be determined by $r_l$, i.e.,
the matrix $\Omega_{f\!f}$ in \eqref{1aef} is nonsingular. 
\end{defn}

\begin{figure}[t]
\centering
\includegraphics[width=0.9\linewidth]{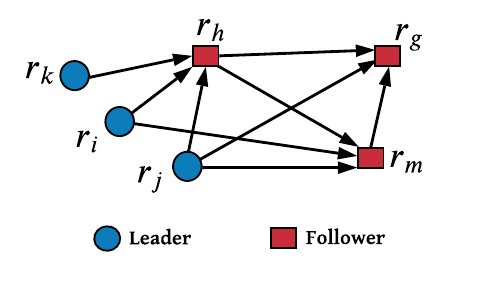}
\caption{Nominal formation. }
\label{lo}
\end{figure}

If $(\mathcal{G}, r)$ is localizable, it is concluded from \eqref{1aef} that
\begin{equation}
r_f = -(\Omega_{f\!f}^{-1}\Omega_{fl}\otimes I_3) r_l.
\end{equation}

An example of how to 
calculate the matrices $\Omega_{fl}$ and $\Omega_{f\!f}$ is shown in Fig. \ref{lo}, where $r_i=(2,0,0)^{\top}$, $r_j=(3,0,0)^{\top}$, $r_k=(1,2,0)^{\top}$, $r_h=(3,3,0)^{\top}$, $r_m=(6,1,0)^{\top}$, $r_g=(7,3,0)^{\top}$.

\begin{enumerate}[(i)]
    \item The displacement constraint among the agents $r_i, r_j, r_k ,r_h$ is
    \begin{equation}\label{o1}
        e_{hi}-\frac{5}{6}e_{hj}- \frac{1}{2}e_{hk}= \mathbf{0}.
    \end{equation}
    \item The displacement constraint among the agents $r_i, r_j, r_h ,r_m$ is
    \begin{equation}
        e_{mi}-\frac{11}{9}e_{mj}-\frac{1}{9}e_{mh}= \mathbf{0}.
    \end{equation}
     \item The displacement constraint among the agents $r_j, r_h, r_m ,r_g$ is
    \begin{equation}\label{o3}
        e_{gj}-\frac{5}{8}e_{gh}-\frac{3}{2}e_{gm}= \mathbf{0}.
    \end{equation}
\end{enumerate}

Let $r_f=(r_h^{\top}, r_m^{\top}, r_g^{\top})^{\top}$ and $r_l=(r_i^{\top}, r_j^{\top}, r_k^{\top})^{\top}$. From \eqref{o1}-\eqref{o3}, it has $ (\Omega_{fl} \otimes I_3)r_l +    (\Omega_{f\!f} \otimes I_3)r_f = \mathbf{0}$, where
\begin{equation}\label{o4}
\Omega_{fl} \!=\! \left[\begin{array}{lll}
   -1 & \ \frac{5}{6} &    \frac{1}{2}   \\
   -1 & \  \frac{11}{9} & 0 \\
   \ \  0 & -1 &   0
    \end{array}\right],  
     \Omega_{f\!f} = \left[\begin{array}{lll}
    -\frac{1}{3} & \ \ 0 & \ \ 0 \\
    \ \ \frac{1}{9} &  -\frac{1}{3} & \ \ 0 \\
    \ \ \frac{5}{8} & \ \ \frac{3}{2} & -\frac{9}{8}
    \end{array}\right].  
\end{equation}

Since $\Omega_{f\!f}$ in \eqref{o4} is nonsingular,
the nominal formation shown in Fig. \ref{lo} is localizable.
Next, the graph condition of a localizable nominal formation is given. The angle-displacement infinitesimal motions are the motions to preserve the invariance of the angle-displacement rigidity matrix, where the proposed follower matrix is a submatrix of the angle-displacement rigidity matrix. The trivial angle-displacement infinitesimal motions include scalings, rotations, and translations of the nominal formation $(\mathcal{G}, r)$.

\begin{defn}\label{infig}
\cite{fang2020} 
A nominal formation $(\mathcal{G}, r)$ is infinitesimally angle-displacement rigid if
all the angle-displacement infinitesimal motions are trivial. 
\end{defn}

Any two nodes are called collocated if the nodes are located at the same position.

\begin{assumption}\label{as1}
{No two nodes in the nominal formation $(\mathcal{G}, r)$ are collocated, and
the nominal formation is infinitesimally angle-displacement rigid}.
\end{assumption}

\begin{remark}\label{ge1}
It is concluded from\cite{fang2020} that $(\mathcal{G}, r)$ is localizable if it is infinitesimally angle-displacement rigid, where the norminal configuration $r$ can be either non-convex or non-generic. From Lemma $5$ and Lemma $6$ in \cite{fang2020},  there must be at least four leaders in $\mathbb{R}^3$ and three leaders in $\mathbb{R}^2$ if $(\mathcal{G}, r)$ is localizable, but there is no constraint on the number of the followers. 
\end{remark}

\begin{remark}
Note that the localizability in terms of the stress matrix is guaranteed by a generic nominal configuration $(\mathcal{G}, r)$ \cite{zhao2018affine}. The advantage of the proposed follower matrix is that
the localizability in terms of the proposed follower matrix $\Omega_f$ can be guaranteed by both generic and non-generic nominal configurations. 
\end{remark}

The positions of the agents in the leader group and follower group are represented by
$p_l \!=\! [p_1^{\top}, \cdots, p_{n_l}^{\top}]^{\top} \! \in \! \mathbb{R}^{3n_l}$ and $p_f \!=\! [p_{n_l\!+\!1}^{\top} , \cdots, p_{n}^{\top}]^{\top} \! \in \! \mathbb{R}^{3n_f}$, respectively. 
The desired time-varying formation $p^*(t)=[{p^*_l}^{\top}\!(t),{p^*_f}^{\top}\!(t)]^{\top}$
is designed as
\begin{equation}\label{ti}
    p^*(t)= a(t)[I_n \otimes Q(t)]g(t)+ {\mathbf{1}}_n \otimes b(t),
\end{equation}
where $a(t) \in \mathbb{R}$, $Q(t) \in SO(3)$, and $b(t) \in \mathbb{R}^3$ represent, respectively, the time-varying scaling, rotational, and translational formation maneuvers. $g(t)= [g_l^{\top}(t), g_f^{\top}(t)]^{\top}$ represents the time-varying formation shape designed by
\begin{equation}\label{vcon1}
(\Omega_{fl}\otimes I_3) g_l(t)+(\Omega_{f\!f} \otimes I_3)g_f(t)  = \mathbf{0},
\end{equation}
where the elements of $\Omega_{fl}$ and $\Omega_{f\!f}$ in \eqref{vcon1} are calculated by the nominal formation $(\mathcal{G}, r)$ shown in \eqref{element}.

\begin{remark}
If it is not required to change the formation shape, $g(t)\!=\!r$, where $r$ is the nominal configuration. Then, the desired time-varying formation in \eqref{ti} becomes
\begin{equation}\label{ttra}
    p^*(t)= a(t)[I_n \otimes Q(t)]r+ {\mathbf{1}}_n \otimes b(t).    
\end{equation}
\end{remark}

\begin{lemma} \label{ll21}
Suppose Assumption \ref{as1} holds \cite{fang2020tc}. For the desired formation $p^*(t)=[{p^*_l}^{\top}\!(t),{p^*_f}^{\top}\!(t)]^{\top}$ in \eqref{ti},  $(\Omega_{f\!f} \otimes I_3 )p^*_f(t) \!+\! (\Omega_{fl} \otimes I_3) p^*_l(t) \!=\! \mathbf{0}$, i.e., $p_f^*(t)$ can be calculated as
\begin{equation}\label{re}
p_f^*(t) = -(\Omega_{f\!f}^{-1}\Omega_{fl}\otimes I_3) p_l^*(t).
\end{equation}
\end{lemma}
\begin{proof}

From \eqref{zer}, it has
\begin{equation}\label{tt2}
    \Omega_{fl} \cdot \mathbf{1}_{n_l} + \Omega_{f\!f} \cdot \mathbf{1}_{n_f} = \mathbf{0}.
\end{equation}

From \eqref{vcon1}, it has
\begin{equation}\label{tt1}
  \begin{array}{ll}
 & [\Omega_{fl} \otimes Q(t)]g_l(t) \!+\! [\Omega_{f\!f} \otimes Q(t)]g_f(t) \\
 &\! =\! [I_{n_f} \! \otimes \! Q(t)] \cdot [ (\Omega_{fl}\otimes I_3) g_l(t)\!+\!(\Omega_{f\!f} \otimes I_3)g_f(t) ] \\
 & \!= \!\mathbf{0}.
 \end{array}
\end{equation}

Combining  \eqref{ti}, \eqref{tt2}, and \eqref{tt1}, it yields that
\begin{equation}\label{ty1}
    \begin{array}{ll}
         & (\Omega_{f\!f} \otimes I_3 )p^*_f(t) + (\Omega_{fl} \otimes I_3) p^*_l(t) \\
         & = a(t) [(\Omega_{fl} \otimes Q(t))g_l(t) \!+\! (\Omega_{f\!f} \otimes Q(t))g_f(t)] \\
         & \ \ \ + [\Omega_{fl} \cdot \mathbf{1}_{n_l} + \Omega_{f\!f} \cdot \mathbf{1}_{n_f} ]\otimes b(t) \\
         &  = \mathbf{0}.
    \end{array}
\end{equation}

Then, the conclusion follows.
\end{proof}

\subsection{Problem Statement}

The $m$th order multi-agent system is described by 
\begin{equation}\label{sysl}
       p_i^{(m)} = u_i,   \  \ m \ge 1, \ i =1, \cdots, n, 
\end{equation}
where $p_i, u_i \in \mathbb{R}^3$ are, respectively, the position and control input of agent $i$ in 3-D space. $m$ is a positive integer, and $p_i^{(m)}$ is the $m$th order derivative of $p_i$. The high-order multi-agent system \eqref{sysl} can be rewritten as 
\begin{equation}\label{sysl1}
\dot x_i = Ax_i+Bu_i, \ \   i = 1, \cdots, n,  
\end{equation}
where $x_i=(p_i^{\top}, [p_i^{(1)}]^{\top}, \cdots, [p_i^{(m-1)}]^{\top})^{\top} \in   \mathbb{R}^{3m  }$ is the motion state of agent $i$, and 
\begin{equation}\label{ab1}
    \begin{array}{ll}
         &   A =  \left[\begin{array}{llll}
    0 & 1 & \cdots & 0 \\
    \vdots  & \vdots & \ddots & \vdots \\
     0 & 0 & \cdots & 1 \\
     0 & 0 & \cdots & 0
    \end{array}\right] \otimes I_3 \in \mathbb{R}^{3m \times 3m}, \\
    \\
    & B = \left[\begin{array}{l}
    0  \\
    \vdots \\
    0 \\
    1
    \end{array}\right] \otimes I_3 \in \mathbb{R}^{3m \times 3}.
    \end{array}
\end{equation}

The measured output $y_i \in \mathbb{R}^q$ of each agent $i$ is given as
\begin{equation}\label{61}
   y_i = Cx_i, \ i =1, \cdots, n, 
\end{equation}
where $C \in \mathbb{R}^{q \times 3m}$. The desired motion state of agent $i$ is denoted by $ x_i^* = [{p_i^*}^{\top}, [{p_i^*}^{(1)}]^{\top}, \cdots, [{p_i^*}^{(m-1)}]^{\top}]^{\top}$.

\begin{remark}
The proposed method can be extended to multi-agent systems with general linear dynamics under some additional conditions, i.e., bounded target positions of the leaders and full row rank of the matrix $B$. 
\end{remark}

\begin{defn}\label{defi}
The multi-agent system \eqref{sysl} is said to achieve the semi-global formation maneuver control if, for any priori given bounded set $\mathcal{X} \in \mathbb{R}^{3m}$,
\begin{equation}\label{zone}
   \lim\limits_{t \rightarrow \infty} (x_i(t) -  x_i^*(t)) = \mathbf{0}, \ \ i \in \mathcal{V}_l \cup \mathcal{V}_f,
\end{equation}
where $x_i(0) \in \mathcal{X}, \ i \in \mathcal{V}_l \cup \mathcal{V}_f$.
\end{defn}

The semi-global stability for multi-agent systems focuses on the feasible solution of closed-loop system based on the given priori bounded set \cite{wang2019semi}. 
Denote $x_l \!=\! [x_1^{\top}, \cdots, x_{n_l}^{\top}]^{\top} $ and $x_f \!=\! [x_{n_l\!+\!1}^{\top} , \cdots, x_{n}^{\top}]^{\top}$ as the motion states of the leaders and followers, respectively. Let $ x_l^* = [[{ x}_1^*]^{\top},  \cdots, [{x}_{n_l}^*]^{\top}]^{\top}$ and $ x_f^* = [[ {x}_{n_l\!+\!1}^*]^{\top},  \cdots, [ {x}_{n}^*]^{\top}]^{\top}$ be the desired motion states of leaders and followers, respectively. From Lemma \ref{ll21},  $x_f^*(t)=-(\Omega_{f\!f}^{-1}\Omega_{fl} \otimes I_{3m})x^*_l(t))$. Hence, under 
Assumption \ref{as1}, 
Definition \ref{defi} is equivalent to the following Definition \ref{defi1}.

\begin{defn}\label{defi1}
The multi-agent system \eqref{sysl} is said to achieve the semi-global formation maneuver control if, for any priori given bounded set $\mathcal{X}_l \in \mathbb{R}^{3mn_l}$,
\begin{align}\label{obj1}
& \lim\limits_{t \rightarrow \infty}(x_l(t) - x_l^*(t))= \mathbf{0}, \\ & \label{obj2}
 \lim\limits_{t \rightarrow \infty}(x_f(t) + (\Omega_{f\!f}^{-1}\Omega_{fl} \otimes I_{3m})x_l(t))=\mathbf{0}, 
\end{align}
where $x_l(0)= [x_1^{\top}(0), \cdots, x_{n_l}^{\top}(0)]^{\top}  \in \mathcal{X}_l$.
\end{defn}

It is concluded from the Lyapunov stability theorem that the stability of the system shown in Section \ref{sefol} can be guaranteed if $\Omega_{f\!f}$ is positive definite. If $\Omega_{f\!f}$ is not positive definite for the nominal positions $r$ and communication graph $\mathcal{G}$, there are two methods to make the matrix $\Omega_{f\!f}$ positive definite. The first method is to change the nominal positions $r$, and the second method is to change the 
communication graph $\mathcal{G}$, i.e., designing a new communication graph $\mathcal{\bar G}$ for the agents. Next, two examples of the second method are presented.
Under Assumption \ref{as1}, i.e., $\Omega_{f\!f}$ is nonsingular, the following follower
matrices $\bar \Omega_f$, $\widehat \Omega_f$ can be used to design controllers of the followers because the matrices $\bar \Omega_{f\!f}= I_{n_f}$ in \eqref{zero1} and  $\widehat \Omega_{f\!f} =  \Omega_{f\!f}^{\top} \Omega_{f\!f}$ in \eqref{zero2} are positive definite.
\begin{align}\label{zero1}
& \bar \Omega_{f}\! = \!
    \Omega_{f\!f}^{-1} \Omega_f \!= \! \left[\begin{array}{cc}
    \Omega_{f\!f}^{-1}\Omega_{fl}  &  I_{n_f} \end{array}\right]. \\ & \label{zero2}
 \widehat{\Omega}_{f}  \!=\!  \Omega_{f\!f}^{\top} \Omega_f \!=\! \left[ \! \begin{array}{cc}
\Omega_{f\!f}^{\top}\Omega_{fl}  &   \Omega_{f\!f}^{\top} \Omega_{f\!f}
    \end{array} \! \right].  
\end{align}

If $\Omega_{f\!f}$ is not positive definite,
it is clear that the communication graph of $\bar \Omega_{f}$ in  \eqref{zero1} or $\widehat{\Omega}_{f}$ in \eqref{zero2} is different from that of $\Omega_{f}$.
The elements of $\bar \Omega_{f}$ and $\widehat{\Omega}_{f}$ also satisfy
\begin{align}\label{fo1}
& \begin{array}{ll}
    &[\bar \Omega_f]_{ij} \!=\!  \left\{ \! \begin{array}{lll} 
    -\bar w_{ij}, & 
    j \in \mathcal{N}_i, \ \ i \neq j, \\
    \ \  0, &  j \notin \mathcal{N}_i, \ \ i \neq j, \\
    \sum\limits_{k \in \mathcal{N}_i}\bar w_{ik},
    & i=j,\\
    \end{array}\right. 
\end{array} \\
& \begin{array}{ll}\label{fo2}
    &[\widehat \Omega_f]_{ij} \!=\!  \left\{ \! \begin{array}{lll} 
    -\widehat w_{ij}, & 
     j \in \mathcal{N}_i, \ \ i \neq j, \\
    \ \  0, &  j \notin \mathcal{N}_i, \ \ i \neq j, \\\sum\limits_{k \in \mathcal{N}_i}\widehat w_{ik},
    & i=j.\\
    \end{array}\right. 
\end{array}
\end{align}

\section{Controllers of the Leaders and Followers}\label{leader}

\subsection{Controllers of the Leaders}

In this paper, the
first leader decides
the time-varying maneuver parameters $a(t), Q(t), b(t), g(t)$ in \eqref{ti}, and the rest leaders calculate their desired positions by  \eqref{ltra} 
once they have obtained 
maneuver parameters $a(t), Q(t), b(t), g(t)$ from the first leader through communication.
\begin{equation}\label{ltra}
      p^*_i(t)= a(t)Q(t)g_i(t)+ b(t), \ \ i \in \{2, \cdots, n_l\},
\end{equation}
where $g_i(t)$ is $i$th element of $g_l(t)$ shown in \eqref{vcon1}.

\begin{assumption}\label{as4} 
The leader $p_i \! \in \! \{2, \cdots, n_l\}$ can obtain information from the first leader $p_1$.
The desired trajectory of each leader is $k$th order differentiable, and its $k$th order derivative ${p^*_i}^{(k)}(t), i \! \in \! \{1, \cdots, n_l\}$
is bounded, i.e., there exists $\gamma_k \in  \mathbb{R}$ satisfying $\|{p^*_i}^{(k)}(t)\|_{2} \! \le \! \gamma_k$, $(k\!=\!1,\cdots, m)$\cite{chen2020distributed}.
\end{assumption}

The controller of leader $i$ is designed as
\begin{equation}\label{vi2}
\begin{array}{ll}
     &  \dot \eta_i = A\eta_i+Bu_i+L(C\eta_i-y_i), \\
     & u_i\!=\! ( \beta \otimes I_3)^{\top} (\eta_i-  x_i^*
)\!+\! {p^*_i}^{(m)}, \ \ i \in \mathcal{V}_l,
\end{array}
\end{equation}
where $\eta_i \in \mathbb{R}^{3m}$ is the motion state estimate of leader $i$.
$\beta = [\beta_0, \beta_1, \cdots, \beta_{m\!-\!1}]^{\top} \in \mathbb{R}^{m} $ is the control gain. $L \in \mathbb{R}^{3m \times q}$
is the output feedback gain matrix. 
The tracking error of leader $i$ is denoted by $e_i = x_i- x_i^*$. The motion state estimation error is denoted by $e_{\eta_i}=\eta_i-x_i$.
Combining \eqref{61} and \eqref{vi2}, it has 
\begin{equation}\label{clo2}
    \left[\begin{array}{l}
    \dot e_i  \\
    \dot e_{\eta_i} 
    \end{array}\right]=  \left[\begin{array}{ll}
    W_1\otimes I_3 & W_2\otimes I_3  \\
    \ \ \ \ \mathbf{0} & A+LC
    \end{array}\right]\left[\begin{array}{l}
     e_i  \\
    e_{\eta_i} 
    \end{array}\right],
\end{equation}
where
\begin{equation}\label{weig}
 W_1 =  \left[\begin{array}{llll}
    0 & 1 &  & 0  \\
    \vdots  &   & \ddots & \\
    0 & 0 & & 1 \\
    \beta_0 & \beta_1 & \cdots & \beta_{m\!-\!1}
    \end{array}\right], 
\end{equation}

and

\begin{equation}
W_2 =  \left[\begin{array}{llll}
    0 & 0 &  & 0  \\
    \vdots  &   & \ddots & \\
    0 & 0 & & 0 \\
    \beta_0 & \beta_1 & \cdots & \beta_{m\!-\!1}
    \end{array}\right].
\end{equation}

For the matrix $W_1$, its eigenvalues $\lambda_1, \lambda_2, \cdots, \lambda_{m}$ satisfy
\begin{equation}\label{asv}
\prod_{k=1}^{m}(\lambda-\lambda_k)= -\beta_0-\beta_1\lambda-\cdots-  \beta_{m\!-\!1} \lambda^{m\!-\!1}+\lambda^{m}.
\end{equation}

Hence, the eigenvalues of $W_1$ can be freely assigned by designing parameters $\beta_0, \beta_1, \cdots, \beta_{m}$ through \eqref{asv}.

\begin{theorem}\label{sth1}
Suppose Assumption \ref{as4} holds. Under controller \eqref{vi2}, for each leader $i \in \mathcal{V}_l$, the controller $u_i$ is bounded and 
the tracking error $e_i$  exponentially converges to zero if
\begin{enumerate}[(i)]
\item The parameters $\beta_0, \beta_1, \cdots, \beta_{m-1}$ are chosen so that
the matrix $W_1$ is Hurwitz;
\item $L= - H^{-1}C^{\top}$, where $H>0$ is a positive definite matrix satisfying the following linear matrix inequality: 
    \begin{equation}\label{thu}
        A^{\top}H +HA-2C^{\top}C <0.
    \end{equation}
\end{enumerate}
\end{theorem}

\begin{proof}
Let
\begin{equation}\label{ww1}
    W= \left[\begin{array}{ll}
    W_1\otimes I_3 & W_2\otimes I_3  \\
    \ \ \ \ \mathbf{0} & A+LC
    \end{array}\right].
\end{equation}

From \eqref{clo2}, it has
\begin{equation}\label{ww2}
\left[\begin{array}{l}
     e_i(t)  \\
    e_{\eta_i}(t) 
    \end{array}\right] = \text{exp}(Wt)\left[\begin{array}{l}
     e_i(0)  \\
    e_{\eta_i}(0) 
    \end{array}\right], \ \  t \ge 0,
\end{equation}
where $\text{exp}(\cdot)$ is the natural exponential function. Since $W_1$ is Hurwitz, there exists a positive definite matrix $\mathcal{M}>0$ satisfying $W_1^{\top}\mathcal{M}+ \mathcal{M}W_1<0$. 
Define $\xi >0$ as a positive constant. For the controller \eqref{vi2},
consider a Lyapunov function \begin{equation}\label{leav}
    V_1 = \sum\limits_{i=1}^{n_l} \left[\begin{array}{ll}
     e_i^{\top} & e_{\eta_i}^{\top}  
    \end{array}\right] \left[\begin{array}{ll}
     \mathcal{M} \otimes I_3 &  \mathbf{0} \\
     \mathbf{0}  &  \xi H
    \end{array}\right] \left[\begin{array}{l}
     e_i  \\
    e_{\eta_i} 
    \end{array}\right].
\end{equation}

Taking the time-derivative of $V_1$ gives
\begin{equation}\label{leeam}
\begin{array}{ll}
     & \dot V_1 = \sum\limits_{i=1}^{n_l} \left[\begin{array}{ll}
     e_i^{\top} & e_{\eta_i}^{\top}  
    \end{array}\right] \mathcal{\bar M} \left[\begin{array}{l}
     e_i  \\
    e_{\eta_i} 
    \end{array}\right],
\end{array}
\end{equation}
where
\begin{equation}\label{leam}
\mathcal{\bar M} =    \left[\begin{array}{ll}
    (W_1^{\top}\mathcal{M}\!+\!\mathcal{M}W_1) \otimes I_3 &   \mathcal{M}W_2 \otimes I_3 \\
      W_2^{\top}\mathcal{M} \otimes I_3   &  \xi(A^{\top}H \!+\!HA\!-\!2C^{\top}C)
    \end{array}\right].
\end{equation}

Note that $A^{\top}H +HA-2C^{\top}C<0$ and $W_1^{\top}\mathcal{M}+ \mathcal{M}W_1<0$. By choosing
$\xi>0$ sufficiently large and using Schur Complement Lemma \cite{boyd1994linear}, it has $\mathcal{\bar M}<0$. Hence, the matrix $W$ is Hurwitz and $\lim\limits_{t \! \rightarrow \! \infty}  e_i(t),  e_{\eta_i}(t)  \rightarrow 0$. From \eqref{ww2},  it is clear that the tracking error of leader $i$ exponentially converges to zero.
In addition, the controller  $u_i$ is bounded, i.e.,
\begin{equation}\label{uup2}
    \| u_i(t) \|_{2} \le \gamma_u, \ \ t\ge 0, i \in \mathcal{V}_l,
\end{equation}
where $\gamma_u$ is the upper bound of the controllers of the leaders.
The details of how to obtain $\gamma_u$ is given in Appendix. 

\end{proof}

\subsection{Controllers of the Followers}\label{sefol}

Inspired by the control algorithms \cite{li2011distributed, cao2010distributed},
the distributed output feedback formation maneuver controller of follower $i \in \mathcal{V}_f$, based on the follower matrix 
$\bar \Omega_f$ in \eqref{zero1}, is designed as
\begin{equation}\label{follower}
\begin{array}{ll}
     & \dot \eta_i \!=\! A \eta_i \!+\! Bu_i \!+\! L(C\eta_i\!-\!y_i) \\
     & \ \ \ \ \  + c_1L\sum\limits_{j=1}^n\bar w_{ij} [(C(\eta_i\!-\!\eta_j) \!-\! (y_i\!-\!y_j)],  \\
     &  u_i \!=\! c_1K \sum\limits_{j=1}^n \bar w_{ij}(\eta_i\!-\!\eta_j)\!+\!c_2\text{sgn}[K \sum\limits_{j=1}^n \bar w_{ij}(\eta_i\!-\!\eta_j)], 
\end{array}    
\end{equation}
where $\eta_i \in \mathbb{R}^{3m}$ is the motion state estimate. The matrices
$A, B,$ and $C$ are given in \eqref{ab1} and \eqref{61}. $\bar w_{ij}  \in \mathbb{R}$ is the edge weight given in \eqref{fo1}. $L \in \mathbb{R}^{3m \times q}$ and $K \in \mathbb{R}^{3 \times 3m}$ are gain matrices. $c_1, c_2 >0 $ are positive constant coupling gains. $\text{sgn}(\cdot)$ is the signum function defined component-wise, i.e., for any vector $\rho = [\rho_1, \cdots, \rho_n]^T \in \mathbb{R}^n$, it has $\text{sgn}(\rho) = [\text{sign}(\rho_1), \cdots, \text{sign}(\rho_n)]^T$, where
$\text{sign}(\rho_i)= \frac{\rho_i}{\|\rho_i\|_2}$ if $\rho_i \neq 0$ and $0$ otherwise. Let $z_i=[x_i^{\top}, \eta_i^{\top}]^{\top}$,
$z_l = [z_1^{\top}, \cdots, z_{n_l}^{\top}]^{\top}$, and $z_f = [z_{n_l\!+\!1}^{\top}, \cdots, z_n^{\top}]^{\top}$. $ u_l=[u_1^{\top}, \cdots, u_{n_l}^{\top}]^{\top}$ is the control inputs of the leaders. Combining \eqref{sysl1}, \eqref{61}, and \eqref{follower}, the closed-loop multi-agent system becomes
\begin{equation}\label{linearf}
    \begin{array}{ll}
         & \dot z_l \!=\! (I_{n_l} \otimes G )z_l+ (I_{n_l} \otimes E )u_l, \\
         & \dot z_f \!= \! (I_{n_f} \otimes G \!+\! c_1I_{n_f} \otimes S )z_f \!+\! c_1(\Omega_{f\!f}^{-1}\Omega_{fl} \otimes S )z_l \\
         & \ \ \ \ \ \ +c_2 (I_{n_f} \otimes E ) D,
    \end{array}
\end{equation}
where
\begin{equation}
\begin{array}{ll}
     &  G = \left[\begin{array}{ll}
    A & \mathbf{0}\\
   -LC  &  A\!+\!LC
    \end{array}\right],  E \!=\! \left[\begin{array}{ll}
   B  \\
    B 
    \end{array}\right], \\
   & S= \left[\begin{array}{ll}
   \ \ \ \mathbf{0} & BK \\
   -LC & LC\!+\!BK
    \end{array}\right], \\
    \\
     &   D \!=\! \text{sgn}[(I_{n_f} \otimes [\mathbf{0} \ K])z_f \!+\! (\Omega_{f\!f}^{-1}\Omega_{fl} \otimes [\mathbf{0} \ K])z_l].
\end{array}    
\end{equation}

Define the tracking error of the followers as
\begin{equation}\label{leff}
    e_f = z_f + (\Omega_{f\!f}^{-1}\Omega_{fl} \otimes I_{6m})z_l.
\end{equation}

It has
\begin{equation}\label{tr1}
\begin{array}{ll}
     & \hspace*{-0.3cm} \dot e_f = \dot z_f + (\Omega_{f\!f}^{-1}\Omega_{fl} \otimes I_{6m}) \dot z_l \\
     & \hspace*{-0.2cm} \ \ \ = (I_{n_f} \otimes G \!+\! c_1I_{n_f} \otimes S)e_f \!+\!\! (\Omega_{f\!f}^{-1}\Omega_{fl} \otimes E)  u_l \\
     & \ \ \ \ + c_2 (I_{n_f} \otimes E ) {D}.
\end{array}
\end{equation}

The right side of \eqref{tr1} is discontinuous because the vector $D$ is discontinuous. Since the signum function $\text{sgn}(\cdot)$ is measurable and locally essentially bounded, the Filippov solution to \eqref{tr1} exists \cite{paden1987calculus}. Hence,
\eqref{tr1} can be rewritten as
\begin{equation}
\begin{array}{ll}
     & \hspace*{-0.3cm} \dot e_f \in^{a.e.}  \mathcal{K}[(I_{n_f} \otimes G \!+\! c_1I_{n_f} \otimes S)e_f \!+\!\! (\Omega_{f\!f}^{-1}\Omega_{fl} \otimes E)  u_l \\
     & \ \ \ \ \ \ \ \  + c_2 (I_{n_f} \otimes E ) {D}],
\end{array}    
\end{equation}
where $a.e.$ represents "almost everywhere".

\begin{theorem}\label{theo1}
Suppose Assumptions \ref{as1}-\ref{as4} hold. Under controller \eqref{follower},  the tracking error $e_f$ \eqref{leff} of the followers asymptotically converges to zero if
\begin{enumerate}

\item $c_1 \!>\! 1$, and $c_2 \! \ge \! n_l \sigma \gamma_u$, where $\sigma \!=\! \max\limits_{i,j}\| \sigma_{ij} \|_2$, $\sigma_{ij}$ is the $(i,j)$-th entry of $\Omega_{f\!f}^{-1} \Omega_{fl}$, and $\gamma_u$ is given in Appendix;
\item $K= - B^{\top}P^{-1}$, where $P> 0$ is a positive definite matrix satisfying the following linear matrix inequality:
    \begin{equation}\label{tee1}
        AP +PA^{\top}-2BB^{\top} <0.
    \end{equation}
\item $L= - H^{-1}C^{\top}$, where $H>0$ is a positive definite matrix satisfying the following linear matrix inequality: 
    \begin{equation}\label{tec}
        A^{\top}H +HA-2C^{\top}C <0.
    \end{equation}
\end{enumerate}
\end{theorem}

\begin{proof}
Let 
\begin{equation}\label{tr}
\bar e_f \!=\! (I_{n_f} \otimes \left[\begin{array}{ll}
   I_{3m} & -I_{3m} \\
    \mathbf{0}  &  \ I_{3m}
    \end{array}\right]) e_f. 
\end{equation}

Then, \eqref{tr1} becomes
\begin{equation}\label{tr2}
\begin{array}{ll}
     & \hspace*{-0.3cm} \dot {\bar e}_f =
   (I_{n_f} \otimes \bar G \!+\! c_1I_{n_f} \otimes \bar S) \bar e_f \!+\!\! (\Omega_{f\!f}^{-1}\Omega_{fl} \otimes \bar E)  u_l \\
     & \ \ \ \ + c_2 (I_{n_f} \otimes \bar E ) {D},
\end{array}
\end{equation}    
where
\begin{equation}
\bar G = \left[\begin{array}{ll}
    A\!+\!LC & \mathbf{0} \\
    -LC  &  A
    \end{array}\right],
\bar S = \left[\begin{array}{ll}
    LC & \mathbf{0} \\
    -LC  &  BK
    \end{array}\right],  \bar E \!=\! \left[\begin{array}{ll}
   \mathbf{0} \\
    B 
    \end{array}\right].   
\end{equation}
    
Consider a Lyapunov function:
\begin{equation}\label{lff}
    V_2 = \frac{1}{2}\bar e_f^{\top}(I_{n_f}\otimes \bar H) \bar e_f,
\end{equation}
where  
\begin{equation}\label{fel}
    \begin{array}{ll}
    & \bar H = \left[\begin{array}{ll}
    \tau H & \mathbf{0} \\
    \mathbf{0} &  P^{-1}
    \end{array}\right]
    \end{array}
\end{equation}
with $\tau$ being a positive constant. Based on the properties of $\mathcal{K[\cdot]}$ \cite{paden1987calculus}, the set-valued Lie derivative of $V_2$ is 
\begin{equation}\label{ve2}
\begin{array}{ll}
      & \dot{\tilde{V}}_2 = \mathcal{K}[\frac{1}{2} \bar e_f^{\top}[ I_{n_f} \otimes ( \bar H \bar G  \!+\!  \bar G^{\top} \bar H  \!+\! c_1 \bar H \bar S \!+\! c_1  \bar S^{\top} \bar H)] \bar e_f \\
     & \ \ \ \ \ \ + \bar e_f^{\top}(\Omega_{f\!f}^{-1}\Omega_{fl} \otimes \bar H \bar E)u_l \!+\! c_2 \bar e_f^{\top}(I_{n_f}\otimes \bar H \bar E){D}]. \\
\end{array}
\end{equation}

Since $D = -\text{sgn}[(I_{n_f}\otimes  \bar E^{\top} \bar H) \bar e_f]$, it has 
\begin{equation}\label{te2}
 c_2 \bar e_f^{\top}(I_{n_f}\otimes \bar H \bar E)\mathcal{D} =  -c_2 \sum\limits_{j=1}^{n_f} \|\bar e_{j\!+\!n_l}^{\top} \bar H \bar E \|_1.
\end{equation}

Based on \eqref{te2} and the fact that  $\mathcal{K}[f]=\{ f\}$ if $f$ is continuous \cite{paden1987calculus}, \eqref{ve2} becomes
\begin{equation}\label{vve2}
\begin{array}{ll}
     &   \dot{\tilde{V}}_2   =  \frac{1}{2} \bar e_f^{\top}[ I_{n_f} \otimes ( \bar H \bar G  \!+\!  \bar G^{\top} \bar H  \!+\! c_1 \bar H \bar S \!+\! c_1  \bar S^{\top} \bar H)] \bar e_f \\
     & \ \ \ \ \ \ + \bar e_f^{\top}(\Omega_{f\!f}^{-1}\Omega_{fl} \otimes \bar H \bar E)u_l \!-\! c_2 \sum\limits_{j=1}^{n_f} \|\bar e_{j\!+\!n_l}^{\top} \bar H \bar E \|_1.
\end{array}
\end{equation}

Let $\bar e_f = [\bar e_{n_l+1}^{\top}, \cdots, \bar e_{n}^{\top}]^{\top}$.
Note that $\|u_i\|_2 \le \gamma_u, i \in \mathcal{V}_l$, where $\gamma_u$ is calculated in Appendix. Then, it yields that
\begin{equation}\label{te1}
\begin{array}{ll}
     & \bar e_f^{\top}(\Omega_{f\!f}^{-1}\Omega_{fl} \otimes \bar H \bar E)u_l \\
     & = \bar e_f^{\top} (I_{n_f}\otimes \bar H \bar E) \left[\begin{array}{ll}
    \sum\limits_{j=1}^{n_l}\sigma_{1j}u_j \\
    \ \ \ \ \ \ \vdots \\
     \sum\limits_{j=1}^{n_l}\sigma_{n_fj}u_j 
    \end{array}\right] \\
    & =  
    \sum\limits_{j=1}^{n_f} \bar e_{j\!+\!n_l}^{\top} \bar H \bar E \sum\limits_{k=1}^{n_l}\sigma_{jk}u_k   \le n_l \sigma \gamma_u 
    \sum\limits_{j=1}^{n_f} \|\bar e_{j\!+\!n_l}^{\top} \bar H \bar E \|_1. 
\end{array}
\end{equation}

Combining \eqref{vve2} and \eqref{te1}, it yields that
\begin{equation}\label{ve3}
\begin{array}{ll}
     &  \dot{\tilde{V}}_2 \le  \frac{1}{2} \bar e_f^{\top}[ I_{n_f} \otimes ( \bar H \bar G  \!+\! \bar G^{\top} \bar H  \!+\! c_1 \bar H \bar S \!+\! c_1 \bar S^{\top} \bar H)] \bar e_f \\
     & \ \ \ \ \ \ + (n_l \sigma \gamma_u\!-\!c_2)
    \sum\limits_{j=1}^{n_f} \|\bar e_{j\!+\!n_l}^{\top} \bar H \bar E\|_1 \\
    & \ \ \ \  \le \frac{1}{2} \bar e_f^{\top}[ I_{n_f} \otimes ( \bar H \bar G  \!+\! \bar G^{\top} \bar H  \!+\! c_1 \bar H \bar S \!+\! c_1 \bar S^{\top} \bar H)] \bar e_f.
\end{array}
\end{equation}

For the term $ \bar H  \bar G  \!+\! \bar G^{\top} \bar H  \!+\! c_1 \bar H \bar S \!+\! c_1 \bar S^{\top} \bar H $ in \eqref{ve3}, it has
\begin{equation}\label{ve4}
\begin{array}{ll}
     &  \left[\begin{array}{ll}
    I_3 &  \mathbf{0} \\
     \mathbf{0}  &  P
    \end{array}\right]( \bar H \bar G  \!+\! \bar G^{\top} \bar H   \!+\! c_1 \bar H \bar S \!+\! c_1 \bar S^{\top} \bar H) \left[\begin{array}{ll}
    I_3 &  \mathbf{0} \\
     \mathbf{0}  &  P
    \end{array}\right] \\
     &  \!=\!  \left[\begin{array}{ll}
    \tau U & -(c_1\!+\!1)C^{\top}L^{\top} \\
    -(c_1\!+\!1)LC  &   AP\!+\!PA^{\top}\!-\!2c_1BB^{\top}
    \end{array}\right], 
\end{array}
\end{equation}
where
\begin{equation}
U \!=\!  HA \!+\!A^{\top}H\!-\!2(c_1\!+\!1)C^{\top}C \le HA\!+\!A^{\top}H\!-\!2C^{\top}C<0.
\end{equation}

It is concluded from \eqref{tee1} that  the matrices $AP\!+\!PA^{\top}\!-\!2c_1BB^{\top}\!<\!0$.
Moreover,
since the matrix $C^{\top}L^{\top}(AP\!+\!PA^{\top}\!-\!2c_1BB^{\top})^{-1} LC$ is a real symmetric matrix, it is diagonalizable, and there must be an orthogonal matrix $F$ satisfying
\begin{equation}
    C^{\top}L^{\top}(AP+PA^{\top}-2c_1BB^{\top})^{-1} LC  = F \Lambda F^{\top},
\end{equation}
where $\Lambda$ is a diagonal matrix. Then, it has
\begin{equation}\label{tee2}
\begin{array}{ll}
     &      \tau U -  (c_1\!+\!1)^2C^{\top}L^{\top}(AP+PA^{\top}-2c_1BB^{\top})^{-1} LC   \\
     & =  F[\tau F^{\top}UF- (c_1+1)^2\Lambda]F^{\top}.
\end{array}
\end{equation} 
where  $\tau \in \mathbb{R}$. It follows \eqref{tee1} and \eqref{tee2} that the parameter $\tau$ can be chosen as
\begin{equation}\label{pa}
 \tau > \frac{(c_1\!+\!1)^2\lambda_{\min}(\Lambda)}{\lambda_{\max}(F^{\top}UF)}. 
\end{equation}

Then, it has $F[\tau F^{\top}UF- (c_1+1)^2\Lambda]F^{\top}<0$. Since $ AP\!+\!PA^{\top}\!-\!2c_1BB^{\top} <0$ \eqref{tee1} and $F[\tau F^{\top}UF- c_1^2\Lambda]F^{\top}<0$ \eqref{tee2}, 
it is clear from Schur Complement Lemma \cite{boyd1994linear} that the matrix ${\bar H} \bar G  \!+\! \bar G^{\top} {\bar H} \!+\! c_1 {\bar H} \bar S\!+\! c_1 \bar S^{\top} {\bar H} < 0$. 
As a result, \eqref{ve3} becomes
\begin{equation}\label{ve5}
\begin{array}{ll}
     & \dot{\tilde{V}}_2 <0, \ \ \text{if} \ \ \bar e_f \neq \mathbf{0}. 
\end{array}
\end{equation}

Note that $ \bar e_f \!=\! (I_{n_f} \otimes \left[\begin{array}{ll}
   I_{3m} & -I_{3m} \\
    \mathbf{0}  &  \ I_{3m}
    \end{array}\right]) e_f \neq \mathbf{0}$ if $e_f \neq \mathbf{0}$. Hence,  $\dot{\tilde{V}}_2 < 0$ if $e_f \neq \mathbf{0}$. From Theorem 3.1 in \cite{shevitz1994lyapunov}, 
the tracking error $ e_f$ of the followers asymptotically converges to zero and the control objective \eqref{obj2} is achieved.
\end{proof}

\begin{remark}
If the matrix $\Omega_{f\!f}$ is designed to be positive definite, each follower can use the edge weight $w_{ij}$ of the follower matrix $\Omega_f$  to achieve formation maneuver control.
\end{remark}

In addition, the distributed output formation maneuver controller of follower $i \in \mathcal{V}_f$ can also be designed based on another follower matrix $\widehat{\Omega}_f$ in \eqref{zero2}, i.e.,
\begin{equation}\label{follower1}
\begin{array}{ll}
     & \dot \eta_i \!=\! A \eta_i \!+\! Bu_i \!+\! L(C\eta_i\!-\!y_i) \\
     & \ \ \ \ \  + c_1L\sum\limits_{j=1}^n \widehat w_{ij} [(C(\eta_i\!-\!\eta_j) \!-\! (y_i\!-\!y_j)],  \\
     &  u_i \!=\! c_1K \sum\limits_{j=1}^n \widehat w_{ij}(\eta_i\!-\!\eta_j)\!+\!c_2\text{sgn}[K \sum\limits_{j=1}^n \widehat w_{ij}(\eta_i\!-\!\eta_j)], 
\end{array}  
\end{equation}
where  $c_1 > \frac{1}{\lambda_{\min}(\Omega_{f\!f}^{\top}\Omega_{f\!f})}$ and the parameters
$ K, L, c_2$ are given in Theorem \ref{theo1}. $\widehat w_{ij}  \in \mathbb{R}$ is the edge weight given in \eqref{fo2}.

\begin{remark}
In controller \eqref{follower} with the follower matrix 
$\bar \Omega_f$ \eqref{zero1}, each follower only obtains information from its neighboring leaders, but in controller \eqref{follower1} with the follower matrix $\widehat{\Omega}_f$ \eqref{zero2}, each follower can not only obtain information from  its neighboring leaders but also its neighboring followers. In applications, the agent dynamics may be more complex. In this case, the proposed formation maneuvering scheme for high-order integrator dynamics can serve as path planner to generate desired trajectories for the low level tracking controller of each agent with complex dynamics \cite{cai2011unmanned}.
\end{remark}

\begin{remark}
Note that the controllers for the followers in \eqref{follower} and \eqref{follower1} rely on the measured output $y_i$. The result of Theorem \ref{theo1} can also be extended to the case when only the relative outputs are available for the followers. In this case, the controller can be modified to the form:
\begin{equation}\label{followerr}
\begin{array}{ll}
     & \dot \eta_{ij} \!=\! A \eta_{ij} \!+\! B(u_i \!-\!u_j) \!+\! L[C\eta_{ij} \!-\! (y_i\!-\!y_j)], \\
     &  u_i \!=\! c_1K \sum\limits_{j=1}^n \bar w_{ij}\eta_{ij}\!+\!c_2\text{sgn}[K \sum\limits_{j=1}^n \bar w_{ij}\eta_{ij}],  \ i \in \mathcal{V}_f,
\end{array}    
\end{equation}
where $\eta_{ij}$ is the relative motion state estimate, and
$y_i\!-\!y_j$ is the relative output.
\end{remark}
 
The proposed output feedback formation maneuver control of high-order multi-agent systems is closed-loop.

\begin{theorem}
Suppose Assumptions \ref{as1}-\ref{as4} hold. The multi-agent system with high-order integrator dynamics achieves semi-global formation maneuver control
under the controllers of the leaders \eqref{vi2} and the controllers of the followers \eqref{follower} or \eqref{follower1} if 
\begin{enumerate}
\item The parameters $\beta_0, \beta_1, \cdots, \beta_{m-1}$ are chosen so that
the matrix $W_1$ is Hurwitz;
\item $c_1 > \max \{1, \frac{1}{\lambda_{\min}(\Omega_{f\!f}^{\top} \Omega_{f\!f})} \}$, and $c_2 \ge n_l \sigma\gamma_u$, where $\sigma = \max\limits_{i,j}\| \sigma_{ij} \|_2$, $\sigma_{ij}$ is the $(i,j)$-th entry of $\Omega_{f\!f}^{-1} \Omega_{fl}$, and $\gamma_u$ is given in Appendix;
    \item $K=-B^{\top}P^{-1}$, where $P>0$ is a positive definite matrix  satisfying the following linear matrix inequality 
    \begin{equation}
        AP +PA^{\top}-2BB^{\top}<0.
    \end{equation}
\item  $L= - H^{-1}C^{\top}$, where $H>0$ is a positive definite matrix satisfying the following linear matrix inequality 
    \begin{equation}\label{hur}
        A^{\top}H +HA-2C^{\top}C <0.
    \end{equation}
\end{enumerate}
\end{theorem}
\begin{proof}
For the controllers \eqref{vi2} and  \eqref{follower},
consider a Lyaponov function \begin{equation}
    V_3 = V_1+ V_2,
\end{equation}
where $V_1$ and $V_2$ are given in \eqref{leav} and \eqref{lff}.  Based on the properties of $\mathcal{K[\cdot]}$ \cite{paden1987calculus}, the set-valued Lie derivative of $V_3$ 
is
\begin{equation}
\begin{array}{ll}
     & {\dot {\tilde V}}_3 \le \sum\limits_{i=1}^{n_l} \left[\begin{array}{ll}
     e_i^{\top} & e_{\eta_i}^{\top}  
    \end{array}\right] \mathcal{\bar M} \left[\begin{array}{l}
     e_i  \\
    e_{\eta_i} 
    \end{array}\right] \\
     & \ \ \ \ + \frac{1}{2} \bar e_f^{\top}[ I_{n_f} \otimes ( {\bar H} \bar G  \!+\! \bar G^{\top} {\bar H} \!+\! c_1 {\bar H} \bar S\!+\! c_1  \bar S^{\top} {\bar H})] \bar e_f,
\end{array}
\end{equation}
where $\mathcal{\bar M}$ is given in \eqref{leam}. From the proof of Theorem \ref{sth1}, 
$\mathcal{\bar M} <0$.
From the proof of Theorem \ref{theo1}, 
${\bar H} \bar G  \!+\! \bar G^{\top} {\bar H} \!+\! c_1 {\bar H} \bar S\!+\! c_1  \bar S^{\top} {\bar H} <0$. 
Then, it yields
    \begin{equation}
       {\dot {\tilde V}}_3 < 0, \ \text{if} \  e_i, e_{\eta_i}  \neq \mathbf{0}, i=1, \cdots, m, \ \bar e_f \neq \mathbf{0}.
    \end{equation}

Note that $ \bar e_f  \neq \mathbf{0}$ if $e_f \neq \mathbf{0}$.
From Theorem 3.1 in \cite{shevitz1994lyapunov}, the tracking errors $ e_i, i=1, \cdots, m$ and $e_f$ asymptotically converge to zero, i.e., the control objectives \eqref{obj1} and \eqref{obj2} are both achieved. Note that $c_2 \ge n_l \sigma\gamma_u$. From Appendix,
 the parameter $\gamma_u$ is determined by the bounded initial 
motion states of the leaders. Thus, from Definition \ref{defi1}, the high-order multi-agent systems achieve semi-global formation maneuver control. Moreover,
the following 
Lyapunov function \eqref{lff1} can be used to 
prove the convergence of the controllers \eqref{vi2} and  \eqref{follower1}. 
\begin{equation}\label{lff1}
\begin{array}{ll}
     &  V_4 = \sum\limits_{i=1}^{n_l} \left[\begin{array}{ll}
     e_i^{\top} & e_{\eta_i}^{\top}  
    \end{array}\right] \left[\begin{array}{ll}
     \mathcal{M} \otimes I_3 &  \mathbf{0} \\
     \mathbf{0}  &  \xi H
    \end{array}\right] \left[\begin{array}{l}
     e_i  \\
    e_{\eta_i} 
    \end{array}\right]  \\
     &  \ \ \ \ \ \  + \frac{1}{2}\bar e_f^{\top}(\Omega_{f\!f}^{\top}\Omega_{f\!f}\otimes \bar H) \bar e_f.
\end{array}
\end{equation}
\end{proof}

\section{Comparison with Existing formation maneuver control of High-order multi-agent systems}\label{discu}

In this section, the proposed method is compared with the existing formation maneuver control of high-order multi-agent systems \cite{chen2020distributed, xu2020affine}.

\begin{defn}\label{def1}
A nominal configuration $r=(r_1^{\top}, \cdots, r_{n}^{\top})^{\top}$ is said to be generic if the coordinates $r_1, \cdots, r_n$ do not satisfy any nontrivial algebraic equation with integer coefficients, e.g., any three nominal agents are non-colinear and any four nominal agents are non-coplanar \cite{zhao2018affine}.
\end{defn}

In practice, the leaders are usually required to start formation maneuver control with bounded initial motion states. Hence, the parameter $\gamma_u$ is known before the formation maneuver control. The edge weights of $\Omega_{f\!f}$, $\Omega_{fl}$ in \eqref{aef} and 
the parameter $\gamma_u$  are constant global information. The maneuver parameters $a(t)$, $Q(t)$, $b(t)$, and $g(t)$ in \eqref{ti} are only available to the leaders. 
Next, five major differences with the existing formation maneuver control of high-order multi-agent systems  \cite{chen2020distributed, xu2020affine} will be presented.

(\romannumeral1) The first difference is the condition. The works in \cite{chen2020distributed, xu2020affine} require
    a affinely localizable
    nominal formation $(\mathcal{G}, r)$, and their desired formation is given by
    \begin{equation}\label{gebe}
      p^*(t)= [I_n \otimes \digamma(t)]r+ {\mathbf{1}}_n \otimes b(t),
\end{equation}
where $\digamma(t)$ is the square matrix. It is proved in \cite{zhao2018affine} that affine localizability of a nominal formation is guaranteed if the nominal configuration $r$ of formation is generic. From the comment of Assumption $2$ in \cite{zhao2018affine}, the affine localizability of a nominal formation may also be guaranteed by non-generic configurations, 
but no theoretical proof is given to verify the relationship between non-generic configurations and localizability. Two cases are given below.

(a) If $\digamma(t)=a(t)Q(t)$, where $a(t) \in \mathbb{R}$ and $Q(t) \in SO(3)$, the shape of their desired formation $p^*(t)$ in \eqref{gebe} remains unchanged, i.e., the formation shapes of $p^*(t)$ and $r$ are the same. \eqref{gebe} becomes 
  \begin{equation}\label{gebe1}
      p^*(t)= a(t)[I_n \otimes Q(t)]r+ {\mathbf{1}}_n \otimes b(t).
\end{equation}

Since their nominal configuration $r$ in \eqref{gebe1} must be designed to be generic to ensure affine localizability,  their desired formation $p^*(t)$ in \eqref{gebe1} must be generic.
Note that the proposed nominal configuration $r$ can be designed to be non-generic by using the proposed rigidity theory \cite{fang2020} and follower matrix.

(b) If $\digamma(t) \neq a(t)Q(t)$, 
the shape of their desired formation $p^*(t)$ in \eqref{gebe} is changed, i.e., the formation shapes of $p^*(t)$ and $r$ are different. Note that they still need a nominal generic configuration $r$ to design the time-varying formation shape, but the proposed time-varying formation shape can be designed based on either generic or non-generic nominal configurations shown
in 
\eqref{ti}, i.e.,
\begin{equation}\label{mti}
    p^*(t)= a(t)[I_n \otimes Q(t)]g(t)+ {\mathbf{1}}_n \otimes b(t),
\end{equation}
where $g(t)= [g_l^{\top}(t), g_f^{\top}(t)]^{\top}$ is the time-varying formation shape.  $g(t)$ is designed by $(\Omega_{fl}\otimes I_3) g_l(t)+(\Omega_{f\!f} \otimes I_3)g_f(t)  = \mathbf{0}$, where the matrices $\Omega_{fl}$ and $\Omega_{f\!f}$ are calculated by either generic or non-generic nominal configurations shown in \eqref{element}.

(\romannumeral2) The second difference is the theory. The controllers in \cite{chen2020distributed, xu2020affine} are designed based on the affine theory, while the controllers in this work are designed based on the proposed rigidity theory \cite{fang2020}. As stated in (\romannumeral1), their theories \cite{chen2020distributed, xu2020affine} can only be applied to generic nominal configurations.

(\romannumeral3) The third difference is the application scenario. Although the affine-theory-based methods \cite{chen2020distributed, xu2020affine} can change the formation shape, the design of formation shape may influence the design of rotation of formation  because the shape and rotation of formation are determined by the same matrix. For example, if it is required to change the shape of formation, part of elements of this matrix will be used to change the shape of formation, and there are insufficient elements of this matrix that can be used to change the rotation of formation. Note that the shape change of formation in the affine-theory-based methods \cite{chen2020distributed, xu2020affine} is usually "shearing and colinearity". Different from \cite{chen2020distributed, xu2020affine}, 
the shape and rotation of formation are tuned by using two independent maneuver parameters $g(t)$ and $Q(t)$. That is,  
the design of formation shape will not influence the design of rotation of formation.  In addition, how to realize different types of formation shapes is studied in \eqref{vcon1}, where the maneuver parameters $g(t)$ represents the target shape of formation. Under Assumption \ref{as1}, given any $g_l(t)$, there exists a solution $g_f(t)$ to \eqref{vcon1}. That is, the formation shape $g(t)$ is determined by the design of $g_l(t)$.

(\romannumeral4) The fourth difference is the measurement.
The works in \cite{chen2020distributed, xu2020affine} assume that relative motion states are available. Although the relative positions may be measured by vision technology, it is difficult to measure the relative velocities and
relative accelerations in real applications, let alone the higher-order relative motion state in higher-order multi-agent systems. Hence,
the proposed output feedback formation maneuver controllers are more practical because it only requires (relative) output measurements instead of (relative) state measurements to achieve formation maneuver control in high-order multi-agent systems. If the relative motion states are available, the proposed controllers in \eqref{follower} can be modified as
\begin{equation}\label{rea}
       u_i \!=\! c_1K \sum\limits_{j=1}^n \bar w_{ij}(x_i\!-\!x_j)\!+\!c_2\text{sgn}[K \sum\limits_{j=1}^n \bar w_{ij}(x_i\!-\!x_j)],
\end{equation}
where the parameters $K, c_1, c_2$ are given in Theorem \eqref{theo1}.
\begin{remark}
In the case that the relative motion states are available, compared with existing works in \cite{chen2020distributed, xu2020affine}, the proposed method still has advantages, i.e., it is a closed-loop formation maneuver control and can be applied to non-generic nominal configurations. Moreover, the global stability of the proposed method can be achieved if the control inputs of the leaders are assumed to have a upper bound. The followers under the controllers \eqref{rea} can form the target time-varying formation by only using relative motion state measurements with respect to their neighbors and do not need to communicate with their neighbors to exchange information.
\end{remark}

(\romannumeral5) The fifth difference is the constraint on the leaders.
In the controllers of the followers,
the works in \cite{chen2020distributed, xu2020affine} assume that the  control inputs and tracking errors of the leaders are zero at any time instant. Then, to design controllers of the followers,
the high-order multi-agent system in \cite{chen2020distributed, xu2020affine} is assumed to be
\begin{equation}\label{open1}
\begin{array}{ll}
     & \dot x_i = Ax_i, \ \  x_i=x_i^*,  \ \   i \in \mathcal{V}_l, \\
     & \dot x_i = Ax_i+Bu_i, \ \   i \in \mathcal{V}_f,
\end{array}
\end{equation}
where  $A, B$ are given in \eqref{ab1}.
That is, the leader group and follower group in \cite{chen2020distributed, xu2020affine} are decoupled as two independent subsystems such that the closed-loop convergence cannot be achieved. {It is intuitive that the control inputs of the leaders cannot be assumed to be zero at any time instant because the leaders need to form the time-varying desired formation by their non-zero control inputs, i.e., the condition $ u_i(t)=0, x_i(t)= x_i^*(t), i \in \mathcal{V}_l$ in \eqref{open1} may not hold.  }
Hence, it is more practical to consider
the proposed closed-loop formation maneuver control shown as
\begin{equation}
\dot x_i = Ax_i+Bu_i, \ \   i \in \mathcal{V}_l \cup \mathcal{V}_f,
\end{equation}
where the control inputs and tracking errors of the leaders are not assumed to be zero at any time instant.

\begin{figure*}[t]
\centering
\includegraphics[width=0.9\linewidth]{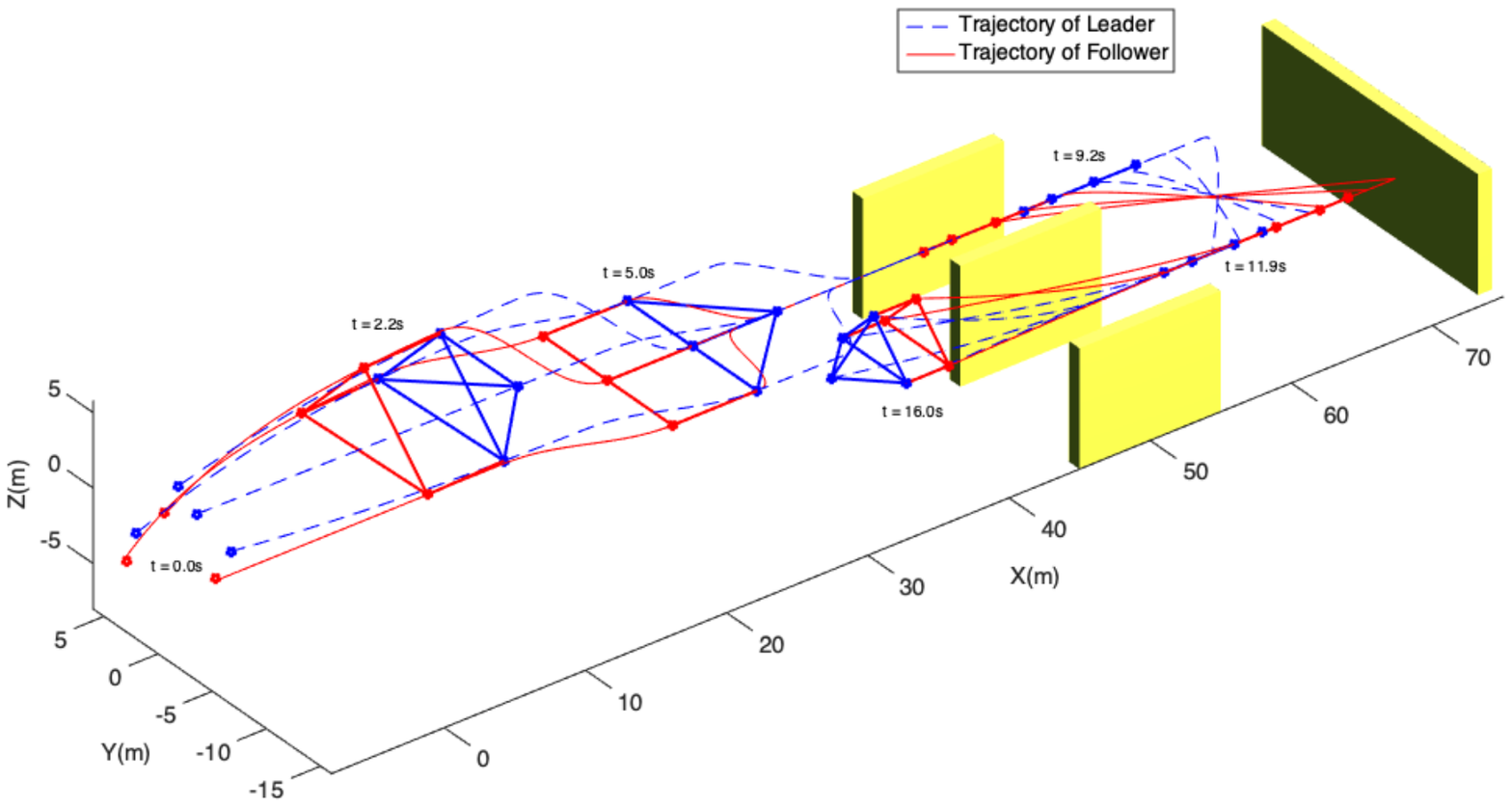}
\caption{ Trajectories of the leaders and followers.}
\label{2d1}
\end{figure*}

\begin{figure}[t]
\centering
\includegraphics[width=0.9\linewidth]{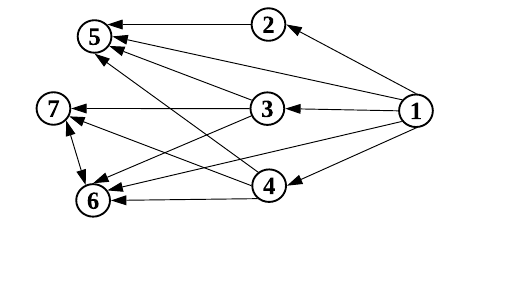}
\caption{Graph condition. }
\label{3d}
\end{figure}

\begin{figure}[t]
\centering
\includegraphics[width=0.9\linewidth]{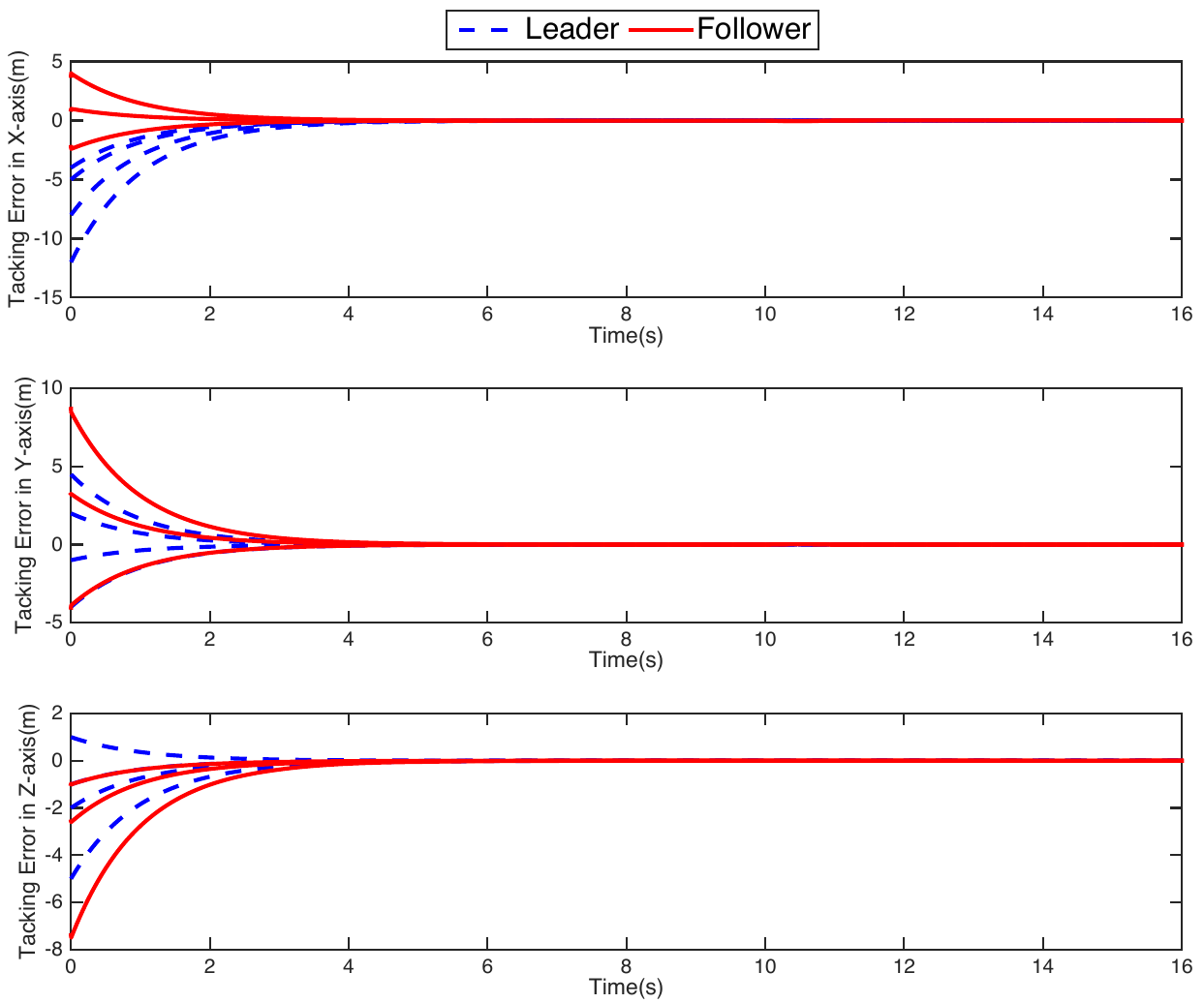}
\caption{Tracking errors.}
\label{2d2}
\end{figure}

Next, some extensions of the proposed method will be introduced.

(a) The followers' controllers \eqref{follower} or \eqref{follower1} can be applied to general linear multi-agent systems under some additional conditions, i.e., bounded target positions of the leaders and full row rank of the matrix $B$.

(b) The proposed method is also applicable to higher-dimensional spaces, i.e., $p_i \! \in \! \mathbb{R}^d, d \!>\! 3$. For example, if $d\!=\!5$, $p_i \!=\!(p_{i1}, p_{i2}, p_{i3}, p_{i4}, p_{i5})^{\top} \in \mathbb{R}^5$, and $u_i=(u_{i1}, u_{i2}, u_{i3}, u_{i4}, u_{i5})^{\top} \in \mathbb{R}^5$, the $m$th order multi-agent system $ p_i^{(m)} = u_i$ in 5-D space can be decoupled as
    two independent $m$th order multi-agent system in 3-D space, i.e.,
    \begin{equation}\label{inde}
        \begin{array}{ll}
             & p_a^{(m)} = u_a, \ \ p_b^{(m)} = u_b, 
        \end{array}
    \end{equation}
where $p_a= (p_{i1}, p_{i2}, p_{i3})^{\top}, p_b= (p_{i4}, p_{i5},0)^{\top}$,    $u_a= (u_{i1}, u_{i2}, u_{i3})^{\top}$, and $u_b= (u_{i4}, u_{i5},0)^{\top}$. Then,
the control objectives of the two independent subsystems shown in \eqref{inde} can be achieved by
the proposed method if the corresponding sub-matrix of $C$ in each subsystem satisfy \eqref{thu}.

\section{Simulation}\label{simu}

A complex environment is simulated in Fig. \ref{2d1}, where the agents are required to pass through several narrow passages. To accomplish this task, the agents need to change the scale, orientation, translation, and shape of formation. In the simulation, only the leaders can sense and respond to the environment, i.e., the leaders determine the  scale, orientation, translation, and shape of formation. Consider a 3-D formation of four leaders $p_l=[p_1^{\top} ,p_2^{\top}, p_3^{\top},p_4^{\top}]^{\top}$ and three followers  $p_f=[p_5^{\top}, p_6^{\top},p_7^{\top}]^{\top}$. 
The nominal configuration $r\!=\! [r_1^{\top} ,r_2^{\top}, r_3^{\top},r_4^{\top},r_5^{\top},r_6^{\top},  r_7^{\top}]^{\top}$ of formation is given by
\begin{equation}
\begin{array}{ll}
     &  r_1 = [6,0,0]^{\top}, \ \ r_2 = [0,0,6]^{\top}, \ \
      r_3 = [0,6,0]^{\top}, \\
      & r_4 = [0,-6,0]^{\top}, \ \ r_5 = [-6,0,6]^{\top}, \ \ r_6 = [-6,6,0]^{\top}, \\
     & r_7 = [-6,-6,0]^{\top}.
\end{array}
\end{equation}

Based on \eqref{element}, the corresponding $\Omega_{fl}, \Omega_{f\!f}$ in \eqref{aef} are
\begin{equation}\label{pao}
    \Omega_{fl} \!=\! \left[\begin{array}{llll}
   -2 & 2 &   \ \ 1 &    \ \ 1   \\
   \ \ 2 & 0 & -3 & -1 \\
    \ \ 0 & 0 & -1 &   \ \ 1 
    \end{array}\right],  
     \Omega_{f\!f} = \left[\begin{array}{lll}
    -2 & 0 & \ \ 0 \\
    \ \ 0 &  2 & \ \ 0 \\
    \ \ 0 & 1 & -1
    \end{array}\right].
\end{equation}

Let $r_l=[r_1^{\top} ,r_2^{\top}, r_3^{\top},r_4^{\top}]^{\top}$ and  $r_f=[r_5^{\top}, r_6^{\top},r_7^{\top}]^{\top}$.
It is easy to prove that $r_f = -(\Omega_{f\!f}^{-1}\Omega_{fl}\otimes I_3)r_l$, i.e., 
the nominal formation is localizable. Then,  the follower matrix $\widehat \Omega_{f}$ is given by
\begin{equation}\label{simf}
\widehat \Omega_{f} \!=\!  \Omega_{f\!f}^{\top} \Omega_f \!=\!  \left[\begin{array}{lllllll}
   4 & -4 & -2 & -2 & 4 & \ \ 0 & \ \ 0      \\
   4 & \ \ 0 & -7 & -1 & 0 & \ \ 5 & -1 \\
   0 & \ \ 0 & \ \ 1 &  -1 & 0 & -1 & \ \ 1
    \end{array}\right].
\end{equation}

In this simulation, based on the follower matrix $\widehat \Omega_{f}$ \eqref{simf}, the leaders and followers implement the controllers
\eqref{vi2} and  \eqref{follower1} over a directed graph shown in Fig. \ref{3d}.
Some maneuver parameters of the simulation are given below.

(\romannumeral1) During the time interval $t \in [3, 5]$,
the scaling, rotational, and translational maneuver parameters are designed as $ a(t) =1,   Q(t) =I_3,  b(t)= [6, 0, 0]^{\top}t$. The formation shape $g(t)=[g_l^{\top}(t), g_f^{\top}(t)]^{\top}$ is designed to satisfy
\begin{equation}\label{vcon}
(\Omega_{fl}\otimes I_3) g_l(t)+(\Omega_{f\!f} \otimes I_3)g_f(t)  = \mathbf{0},
\end{equation}
where $\Omega_{fl}$ and $\Omega_{f\!f}$ are
given in \eqref{pao}. To make all agents coplanar,  $g_l(t)$ is given by
\begin{equation}\label{po1}
  g_l(t) = \left[\begin{array}{llllllllllll}
   6 \! & \! 0 \!&\! -2 \!&\! 0 \!&\! 0 \!&\! -2 \!&\! 0 \!&\! 6 \!&\! -2 \!&\! 0 \!&\! -6 \!&\! -2 
\end{array}\right]^{\top}.
\end{equation}

Combining \eqref{pao} and \eqref{po1}, it has
\begin{equation}\label{po2}
\begin{array}{ll}
     &  g_f(t) \!=\! -(\Omega_{f\!f}^{-1}\Omega_{fl}\otimes I_3)g_l(t)\\
     & \ \ \ \ \ \ \ \!=\! \left[\begin{array}{lllllllll}
    -6 \!&\! 0 \!&\!  -2  \!&\! -6 \!&\! 6 \!&\! -2 \!&\! -6 \!&\! -6 \!&\! -2 
    \end{array}\right]^{\top}.
\end{array}
\end{equation}

\begin{figure}[t]
\centering
\includegraphics[width=0.9\linewidth]{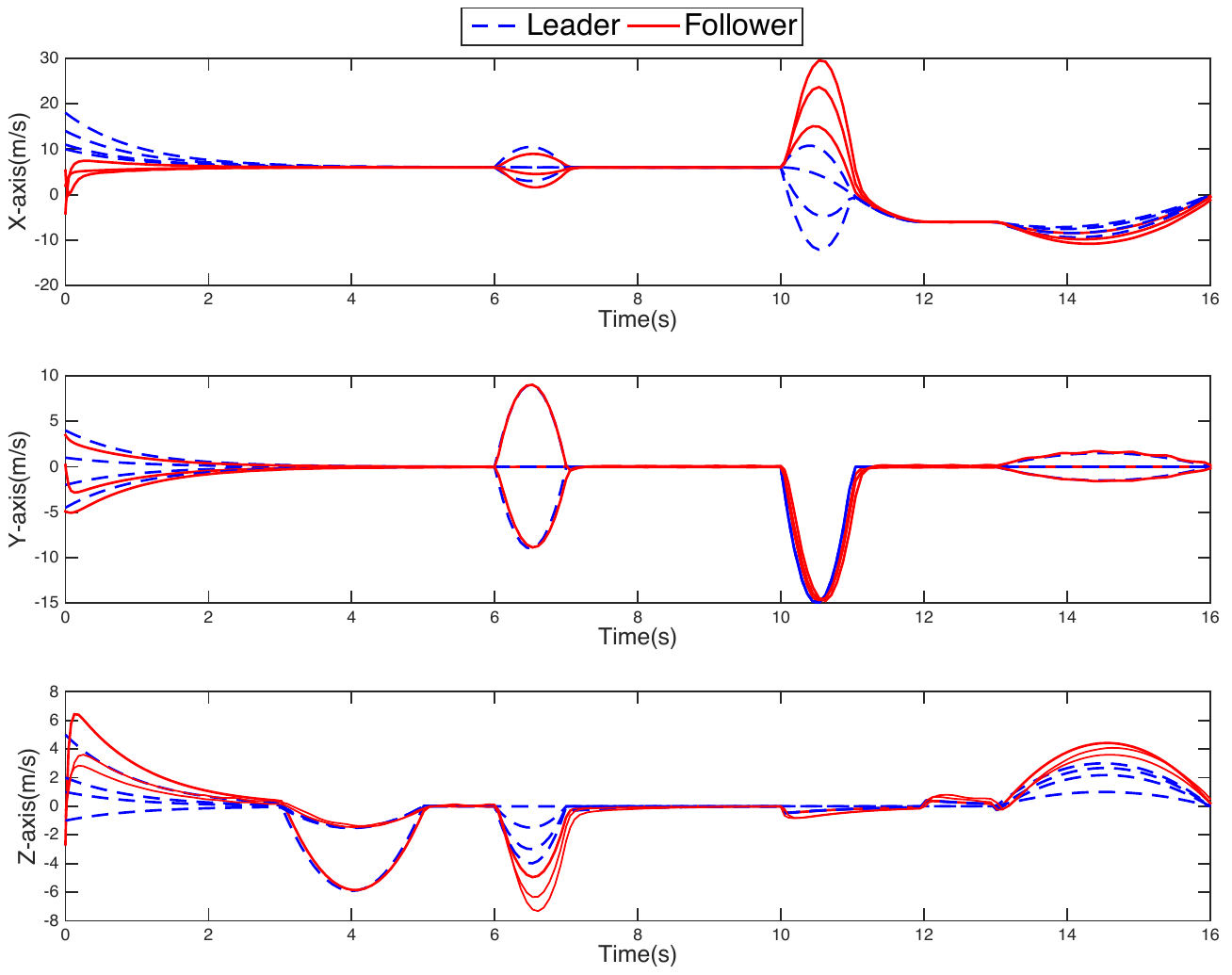}
\caption{Control inputs.}
\label{2d2e}
\end{figure}

Then, the desired formation $p^*(t)$ in \eqref{ti} becomes
\begin{equation}\label{ta1}
 p^*(t)= g(t)+{\mathbf{1}}_7 \otimes \left[\begin{array}{l}
    6 \\
     0 \\
     0
    \end{array}\right]t , \ t \in [3, 5].   
\end{equation}

It is shown in Fig. \ref{2d1} that all the agents are coplanar during the time interval $t \in [3, 5]$.

(\romannumeral2) During the time interval $t \in [7, 10]$, 
the agents pass through the first narrow passage by making all agents colinear. The scaling, rotational, and translational maneuver parameters are designed as $ a(t) =1,   Q(t) =I_3,  b(t)= [6, 0, 0]^{\top}t$. The formation shape $g(t)=[g_l^{\top}(t), g_f^{\top}(t)]$ is also designed satisfying \eqref{vcon}, i.e,
\begin{equation}
 g_l(t) = \left[\begin{array}{llllllllllll}
   6 \! & \! 0 \!&\! -2 \!&\! 3 \!&\! 0 \!&\! -2 \!&\! 0 \!&\! 0 \!&\! -2 \!&\! -2 \!&\! 0 \!&\! -2 
\end{array}\right]^{\top}, 
\end{equation}
and
\begin{equation}
\begin{array}{ll}
     &  g_f(t) \!=\! -(\Omega_{f\!f}^{-1}\Omega_{fl}\otimes I_3)g_l(t)\\
     & \ \ \ \ \ \ \ \!=\! \left[\begin{array}{lllllllll}
    -4 \!&\! 0 \!&\!  -2  \!&\! -7 \!&\! 0 \!&\! -2 \!&\! -9 \!&\! 0 \!&\! -2
    \end{array}\right]^{\top}. 
\end{array}
\end{equation}

Then, the desired formation $p^*(t)$ in \eqref{ti} becomes
\begin{equation}\label{ta2}
 p^*(t)= g(t)+{\mathbf{1}}_7 \otimes \left[\begin{array}{l}
    6 \\
     0 \\
     0
    \end{array}\right]t , \ t \in [7, 10].   
\end{equation}

It is shown in Fig. \ref{2d1} that all the agents are colinear during the time interval $t \in [7, 10]$.
Similarly, the agents can pass through the second narrow passage by designing the maneuver parameters $a(t)$, $Q(t)$, $b(t)$, and $g(t)$ during the time interval $t \in [10, 16]$. 
It is clear in
Fig. \ref{2d2} that the tracking errors asymptotically converge to zero, where the corresponding control inputs are given in Fig. \ref{2d2e}. More information about the simulation can be found in \url{https://youtu.be/DlW6V_lDqM0}. 

\begin{remark}
Although the inter-agent collision avoidance is not considered, the maneuver parameters $g(t)$, $a(t)$, $Q(t),$ $b(t)$ can be designed based on the initial tracking errors of the agents to avoid inter-agent collision as shown in \cite{zhao2019bearing12}.
\end{remark}

\section{Conclusion}\label{conc}

In this article, the high-order formation maneuver control problem is addressed, where the  maneuver parameters are only known to the leaders. The proposed algorithm does not need to design estimators for the followers to estimate the maneuver parameters only known to the leaders. The nominal configuration of formation can be either generic or non-generic,
and 
the tracking errors asymptotically converge to zero. Moreover, 
the proposed method shows how to realize different types of formation shape, which can also be extended to general linear multi-agent systems under some conditions.

\section*{Appendix}

The matrix $W \in \mathbb{R}^{6m \times 6m}$ in \eqref{ww1} can be expressed as 
\begin{equation}\label{pat}
W =  M J M^{-1},      
\end{equation}
where $J$ is the Jordan form of $W$ and $M$ is a nonsingular matrix.
Then, it has
\begin{equation}\label{ww4}
    \text{exp}(Wt) \!=\! M \text{exp}(Jt) M^{-1}.
\end{equation}

Since the matrix $W$ is Hurwitz, the eigenvalues of $J$ are negative. From the matrix theory \cite{horn2012matrix}, it has
\begin{equation}\label{ww3}
 \|\text{exp}(Jt)\|_2 \le  \psi, \ t>0,
\end{equation}
where 
\begin{equation}
\psi =  \sqrt{6m + 3m(6m-1)(\frac{6m-1}{\|\lambda_{\max}(W)\|_2 \text{exp}(1)})^{6m-1}}. 
\end{equation}

If the eigenvalues of matrix $W$ are designed to satisfy $\frac{6m-1}{\|\lambda_{\max}(W)\|_2 \text{exp}(1)}<1$, it has
\begin{equation}
 \psi <  \sqrt{18m^2+ 3m}.
\end{equation}

In addition, if the negative eigenvalues of matrix $W$ are designed to be distinct, the matrix $W$ is diagonalizable, i.e, the matrix $J$ in \eqref{pat} is a diagonal matrix. Then, it has $\|\text{exp}(Jt)\|_2 = \text{exp}(\lambda_{\max}(W)t) < 1$. For this special case, the parameter $\psi=1$. 
Combining \eqref{ww2}, \eqref{ww4}, and \eqref{ww3}, 
it yields that
\begin{equation}\label{by}
\| e_i(t) \|_2, \| e_{\eta_i}(t) \|_2  \le \psi\|M\|_2\|M^{-1}\|_2 (\| e_i(0) \|_2 \!+\! \| e_{\eta_i}(0) \|_2).  
\end{equation}

From \eqref{vi2} and \eqref{by}, it has
\begin{equation}\label{bou1}
 \|  u_i(t) \|_2 \le  2\psi \|\beta\|_2 \|M\|_2\|M^{-1}\|_2  (\| e_i(0) \|_2 \!+\! \| e_{\eta_i}(0) \|_2)  \!+\! \gamma_{m},  
\end{equation}
where $\gamma_{m}$ is defined in Assumption \ref{as4}.
In practice, the leaders are usually required to start formation maneuver control with bounded initial motion states, i.e., the upper bound of $e_i(0)$ and $e_{\eta_i}(0)$ is known before the implementation of the formation maneuver control. Let $\zeta$ be the upper bound of $e_i(0)$ and $e_{\eta_i}(0)$, i.e.,
\begin{equation}\label{lu}
 \| e_i(0)\|_2, \| e_{\eta_i}(0) \|_2 \le \zeta, \ \ i \in \mathcal{V}_l.  
\end{equation}

From \eqref{bou1} and \eqref{lu}, it has
\begin{equation}
  \|  u_i(t) \|_2 \le 4 \psi   \zeta \|\beta\|_2 \|M\|_2\|M^{-1}\|_2     \!+\! \gamma_{m}, \ \ t \ge 0.
\end{equation}

Then, define the upper bound of $ u_i(t), i\in \mathcal{V}_l$ as 
\begin{equation}\label{up}
\gamma_u=  4 \psi \zeta \|\beta\|_2 \|M\|_2\|M^{-1}\|_2      \!+\! \gamma_{m}. 
\end{equation}

As stated above, the parameter $\zeta$ is known if the leaders start formation maneuver control with bounded initial motion states, i.e., 
the upper bound $\gamma_u$ in \eqref{up} is also known if the leaders start formation maneuver control with bounded initial motion states.

\ifCLASSOPTIONcaptionsoff
  \newpage
\fi

\balance
\bibliographystyle{IEEEtran}
\bibliography{papers}

\end{document}